\documentclass[11pt]{article}
\pdfoutput=1
\usepackage{amssymb}
\usepackage{amsmath}
\usepackage{mathtools}
\usepackage{amsfonts}
\usepackage{amsthm}
\usepackage{dsfont}
\usepackage{mathpazo}
\usepackage{fullpage}
\usepackage{enumerate}
\usepackage{subcaption}
\usepackage[
	colorlinks=true,
	linkcolor=blue,
	citecolor=blue,
	urlcolor=blue
	]{hyperref}
\usepackage{tcolorbox}
\usepackage[normalem]{ulem}
\usepackage{authblk}
\usepackage{enumitem}

\usepackage{graphicx}
\usepackage{epstopdf}

\usepackage{doi}
\usepackage[numbers, sort&compress]{natbib}

\newcommand\op{\operatorname}

\newcommand{\ip}[2]{\langle #1,#2\rangle}
\newcommand{\bigip}[2]{\big\langle #1,#2\big\rangle}

\newcommand{\complex}{\mathbb{C}}

\newcommand{\calX}{\mathcal{X}}
\newcommand{\calY}{\mathcal{Y}} 
\newcommand{\calZ}{\mathcal{Z}}

\newcommand{\comp}{\mathrm{C}}
\newcommand{\jord}{\mathrm{J}} 

\newcommand{\Tr}{\op{Tr}}
\newcommand{\I}{\mathds{1}} 
\renewcommand{\t}{T} 
\newcommand{\tX}{\t_{\calX}} 
\renewcommand{\dim}{\op{dim}} 

\renewcommand\L{\op{L}}
\newcommand\Pos{\op{Pos}}
\newcommand\Herm{\op{Herm}}
\newcommand\C{\op{C}}
\newcommand\D{\op{D}}
\newcommand\U{\op{U}} 
\newcommand\T{\op{T}} 
\newcommand{\JP}{\op{JP}} 
\newcommand{\Sep}{\op{Sep}} 
\newcommand{\PPT}{\op{PPT}}

\theoremstyle{plain}
\newtheorem{theorem}{Theorem}
\newtheorem{lemma}[theorem]{Lemma}
\newtheorem{corollary}[theorem]{Corollary}
\newtheorem{proposition}[theorem]{Proposition}
\newtheorem{result}[theorem]{Result}

\theoremstyle{definition}
\newtheorem{definition}[theorem]{Definition}
\newtheorem{remark}[theorem]{Remark}
\newtheorem{example}[theorem]{Example} 

\title{
Jordan products of quantum channels and their compatibility\\ \quad
}

\author{
Mark Girard\thanks{ 
Institute for Quantum Computing, University of Waterloo, Ontario, Canada. \url{mark.girard@uwaterloo.ca}
} 
\quad \quad \quad 
Martin Pl\'{a}vala\thanks{
Naturwissenschaftlich-Technische Fakult\"{a}t Universit\"{a}t Siegen, Siegen, Germany.
Mathematical Institute, Slovak Academy of Sciences, 
Bratislava, Slovakia.
\url{martin.plavala@uni-siegen.de}
} 
\quad \quad \quad 
Jamie Sikora\thanks{
Virginia Polytechnic Institute and State University, Blacksburg, Virginia, USA. 
Institute for Quantum Computing, University of Waterloo, Waterloo, Ontario, Canada. 
Perimeter Institute for Theoretical Physics, Waterloo, Ontario, Canada. 
\url{sikora@vt.edu}
}
}

\date{September 7, 2020}

\begin{document}

\maketitle

\begin{abstract}
Given two quantum channels, we examine the task of determining whether they are com\-patible---meaning that one can perform both channels simultaneously but, in the future, choose exactly one channel whose output is desired (while forfeiting the output of the other channel). We show several results concerning this task. First, we show it is equivalent to the quantum state marginal problem, i.e., every quantum state marginal problem can be recast as the compatibility of two channels, and vice versa. Second, we show that compatible measure-and-prepare channels (i.e., entanglement-breaking channels) do not necessarily have a measure-and-prepare compatibilizing channel. Third, we extend the notion of the Jordan product of matrices to quantum channels and present sufficient conditions for channel compatibility. These Jordan products and their generalizations might be of independent interest. Last, we formulate the different notions of compatibility as semidefinite programs and numerically test when families of partially dephasing-depolaring channels are compatible.
\end{abstract}

%%%%%%%%%%%%%%%%%%%%%%%%%%%%%%%%%%%%%%%%%%%%%%%%%%%%%% 

\section{Introduction}
We first introduce the different settings for the compatibility of states, measurements, and channels that are considered in this paper. (The interested reader is referred to the reviews in references~\cite{Lahti-coexistence, HeinosaariMiyaderaZiman-review, KiukasLahtiPellonpaaKari-complementarity, Heinosaari-review} for further discussions on these topics.) This overview is followed by a summary of our main results.

%%%%%%%%%%%%%%%%%%%%%%%%%%% 

\subsection{Overview of compatibility problems in quantum information theory}

\subsubsection*{The quantum state marginal problem.}
The \emph{quantum state marginal problem}~\cite{Klyachko-marginalProblem, Klyachko-fermionMarginalProblem, WyderkaHuberGuhne-marginalProblem} is one of the most fundamental problems in quantum theory and quantum chemistry. One version of this problem is the following problem. Given a collection of systems $\calX_1, \ldots, \calX_n$ and a collection of density operators $\rho_1, \ldots, \rho_m$---each acting on some respective subset of subsystems $S_1, \ldots, S_m \subseteq \calX_1 \otimes \cdots \otimes \calX_n$---determine whether there exists a state $\rho$ on $\calX_1 \otimes \cdots \otimes \calX_n$ which is consistent with every density operator $\rho_1, \ldots, \rho_m$. For example, if $\rho_1$ acts on $S_1$, then $\rho$ must satisfy
\begin{equation}
\Tr_{\calX_1 \otimes \cdots \otimes \calX_n \setminus S_1}(\rho) = \rho_1,
\end{equation}
and similarly for the other states $\rho_2, \ldots, \rho_m$. This problem is nontrivial if the density operators act on overlapping systems and indeed is computationally expensive to determine, as the problem is known to be QMA-complete \cite{Liu-QMA, BroadbentGrilo-QMA}. Small instances of the problem can be solved (to some level of numerical precision), however, by solving the following semidefinite programming feasibility problem
\begin{equation}
\begin{split}
\text{find:} \quad & \rho \in \Pos(\calX_1 \otimes \cdots \otimes \calX_n)\\
\text{satisfying:}\quad &
\Tr_{\calX_1 \otimes \cdots \otimes \calX_n \setminus S_1}(\rho) = \rho_1 \\
& \quad \quad \quad \quad \vdots \\
& \Tr_{\calX_1 \otimes \cdots \otimes \calX_n \setminus S_m}(\rho) = \rho_m.
\end{split}
\end{equation}
(See Section~\ref{section:background} for a detailed description of this notation.) Note that the condition that $\rho$ has unit trace is already enforced by the constraints.
In the case where the systems $\calX_2 = \cdots=\calX_n$ and density operators $\sigma:=\rho_2 = \cdots = \rho_n$ are identical (where we omit $\rho_1$ for indexing convenience), we remark that, for the choice of subsystems $S_i = \calX_1\otimes \calX_i$ for each $i\in\{2,\dots,n\}$, one obtains the so-called $(n-1)$-symmetric-extendibility condition for $\sigma$. This is closely related to separability testing~\cite{DPS}.

%%%%%%%%%%%%%%%%%%%%%%%%%%% 

\subsubsection*{The measurement compatibility problem.}
There is an analogous task for quantum measurements called the \emph{measurement compatibility problem} (see~\cite{JaeKyunghyunJungheeJinhyoung-negativity, DesignolleFarkasKaniewski-jointMeas} for POVMS and~\cite{Busch-qubit, BuschHeinosaari-jointMeasQubits} for the special case of qubits). This task can be stated as follows. Two POVMs $\{ M_1, \ldots, M_n \}$ and $\{ N_1, \ldots, N_n \}$ are said to be \emph{compatible} if there exists a choice of POVM $\{ P_{i,j} : i \in \{ 1, \ldots, n \}, j \in \{ 1, \ldots, m \} \}$ that satisfies
\begin{equation}
M_i = \sum_{j=1}^m P_{i,j} \quad \text{ and } \quad N_j = \sum_{i=1}^n P_{i,j}
\end{equation}
for each index $i$ and $j$. In other words, the two measurements are a course-graining of the compatibilizing measurement $\{ P_{i,j} \}$. It may seem at first glance that both measurements are being performed simultaneously, but this view is incorrect. Performing a compatibilizing measurement should be viewed as performing a separate measurement which captures the probabilities of both measurements simultaneously. Although measurement compatibility is defined mathematically, it also has operational applications---for example, incompatible measurements are necessary for quantum steering \cite{UolaMoroderGuhne-steering, QuintinoVertesiBrunner-steering, UolaCostaNguyenGuhne-steering} and Bell non-locality \cite{WolfPerezgarciaFernandez-Bell, BrunnerCavalcantiPironioScaraniWehner-Bell, RossetBancalGisin-Bell}.

Determining the compatibility of the two POVMs above can be solved via the following semi\-{}definite programming feasibility problem
\begin{equation}
\label{eq:PMN_sdp}
\begin{split}
\text{find:} \quad & P_{i,j} \in \Pos(\calX), \text{ for } i \in \{ 1, \ldots, n \}, j \in \{ 1, \ldots, m \} \\
\text{satisfying:}\quad &
M_i = \sum_{j=1}^m P_{i,j}, \text{ for each } i \in \{ 1, \ldots, n \} \\ & N_j = \sum_{i=1}^n P_{i,j}, \text{ for each } j \in \{ 1, \ldots, m \},
\end{split}
\end{equation}
where the POVMs each act on some system $\calX$. Note that the condition $\sum_{i=1}^n \sum_{j=1}^m P_{i,j} = \I_{\calX}$ is enforced by the constraints and thus every collection of operators satisfying the conditions in \eqref{eq:PMN_sdp} necessarily composes a POVM. The semidefinite programming formulation provided above may also be used to certify incompatibility by using notions of duality, i.e., to provide an incompatibility witness \cite{CarmeliHeinosaariToigo-witness, Jencova-witness}.

This notion of compatibility for measurements generalizes the concept for commuting measurements. Indeed, if $[M_i, N_j] = 0$ holds for each choice of indices $i$ and $j$ then the measurement operators $P_{i,j}$ defined as $P_{i,j} = M_i N_j$ form a compatibilizing POVM. This does not hold generally, as the operators $M_i N_j$ need not even be Hermitian if $M_i$ and $N_j$ do not commute. Nevertheless, one can study Hermitian versions of these matrices using Jordan products, as is discussed below.

%%%%%%%%%%%%%%%%%%%%%%%%%%% 

\subsubsection*{Jordan products of matrices.}

The \emph{Jordan product} of two square operators $A$ and $B$ is defined as
\begin{equation}
\label{jmat}
A \odot B = \frac{1}{2}(AB + BA) .
\end{equation}
This is Hermitian whenever both $A$ and $B$ are Hermitian. The Jordan product can be used to study the compatibility of measurements (see~\cite{Heinosaari-coexistence, HeinosaariJivulescuNechita-randomPOVMs}). In particular, for POVMs $\{ M_1, \ldots, M_n \}$ and $\{ N_1, \ldots, N_m \}$, note that the operators defined as
\begin{equation}
\label{compmeas}
P_{i,j} := M_i \odot N_j
\end{equation}
satisfy
\begin{equation}
M_i = \sum_{j=1}^m P_{i,j} \quad \text{ and } \quad N_j = \sum_{i=1}^n P_{i,j}
\end{equation}
for each choice of indices $i$ and $j$. Thus, if $M_i \odot N_j$ is positive semidefinite for each $i$ and $j$ then the POVM defined in Equation~\eqref{compmeas} is a compatibilizing measurement.

%%%%%%%%%%%%%%%%%%%%%%%%%%% 

\subsubsection*{The no-broadcasting theorem.}

Before discussing the quantum channel marginal problem, we take a slight detour and discuss the no-broadcasting theorem. A \emph{quantum broadcaster} for a quantum state $\sigma \in \Pos(\calX \otimes \calY)$ is a channel that acts on the $\calX$ subsystem of $\sigma$ and outputs a state $\rho \in \Pos(\calX_1 \otimes \calX_2 \otimes \calY)$ that satisfies
\begin{equation}
\Tr_{\calX_1}(\rho) = \sigma
\quad \text{ and } \quad
\Tr_{\calX_2}(\rho_1) = \sigma,
\end{equation}
where $\calX = \calX_1 = \calX_2$. In other words, one applies the channel that broadcasts $\sigma$ and decides afterwards where $\sigma$ is to be localized. One can show---when there is no promise on the input $\sigma$---that such a channel cannot exist. An easy way to see this is by trying to broadcast half of an EPR state, which would violate monogamy of entanglement. This proves the well-known \emph{no-broadcasting theorem}, which states that a perfect broadcasting channel cannot exist. (The question of determining the best ``noisy'' broadcasting channel has also been investigated~\cite{BuzekHillery-cloning, Werner-cloning, KeylWerner-cloning}.) One may notice that the conditions above imposed by broadcasting is a special case of the quantum state marginal problem (for fixed input $\sigma$) and, moreover, of the symmetric extendibility problem as described above (for which it is sometimes the case that no solution exists).

%%%%%%%%%%%%%%%%%%%%%%%%%%% 

\subsubsection*{The quantum channel marginal problem/channel compatibility problem.}

The task of determining \emph{compatibility of quantum channels}, which has been studied recently (see, e.g.,~\cite{HeinosaariMiyadera-channels, Plavala-channels, Kuramochi-channels, Kuramochi-ETBchannels, Haapasalo-compat}) can be stated as follows. Given two quantum channels, $\Phi_1$ from $\calX$ to $\calY_1$ and $\Phi_2$ from $\calX$ to $\calY_2$, one determines if there exists another channel $\Phi$ from $\calX$ to $\calY_1 \otimes \calY_2$ that satisfies
\begin{equation}
\Phi_1(X) = \Tr_{\calY_2}(\Phi(X))
\quad \text{ and } \quad
\Phi_2(X) = \Tr_{\calY_1}(\Phi(X))
\end{equation}
for every input $X$. The notion of broadcasting (as defined in the previous paragraph) is a specific instance of this problem, where one chooses both $\Phi_1$ and $\Phi_2$ to be the identity channel. In this work we show how channel compatibility can be expressed as a semidefinite programming feasibility problem, but for now we may express it as a convex feasibility problem over the space of linear maps:
\begin{equation}
\begin{split}
\text{find:} \quad & \Phi \text{ completely positive} \\
\text{satisfying:}\quad &
\Phi_1 = \Tr_{\calY_2} \circ \; \Phi \\
& \Phi_2 = \Tr_{\calY_1} \circ \; \Phi.
\end{split}
\end{equation}
The condition that $\Phi$ is trace preserving follows from the constraints.

%%%%%%%%%%%%%%%%%%%%%%%%%%% 

\subsubsection*{Connections between marginal problems.}

To see how the channel and state versions of the marginal problem are generalizations of each other, we may consider the Choi representations of the channels. It is not hard to see \cite{Plavala-channels, HaapasaloKraftMiklinUola-marginalProblem} that the conditions for the compatibility of $\Phi_1$ and $\Phi_2$ is equivalent to the following conditions on the Choi matrices:
\begin{equation}
\Tr_{\calY_2}(J(\Phi)) = J(\Phi_1)
\quad \text{ and } \quad
\Tr_{\calY_1}(J(\Phi)) = J(\Phi_2),
\end{equation}
where $J(\Phi_1)$ and $J(\Phi_2)$ are the Choi representations of these channels. In other words, the channels are compatible if and only if the normalized Choi matrices are compatible as states.

It is not obvious that the channel and state marginal problems are equivalent. Indeed, normalized Choi matrices only constitute a special case of quantum states---they must satisfy certain partial trace conditions. One of the results proved in this paper states that the quantum channel marginal problem also generalizes the state version. Thus these problems can be viewed as equivalent.

%%%%%%%%%%%%%%%%%%%%%%%%%%%

\subsubsection*{The quantum channel marginal problem generalizes the measurement version.}

To see how the channel compatibility problem is a generalization of the measurement compatibility problem, consider the following reduction. For the POVMs $\{ M_1, \ldots, M_n \}$ and $\{ N_1, \ldots, N_m \}$, define the channels
\begin{equation}
\Phi_M(X) = \sum_{i=1}^n \ip{M_i}{X} E_{i,i}
\quad \text{ and } \quad
\Phi_N(X) = \sum_{j=1}^m \ip{N_j}{X} E_{j,j}
\end{equation}
where $E_{i,i}$ is the density matrix of the $i$\textsuperscript{th} computational basis state.\footnote{One may also think of $E_{i,i}$ as $| i \rangle \! \langle i |$ in Dirac notation.} Channels of this form are known as \emph{measure-and-prepare} (or, equivalently, as \emph{entanglement-breaking}) channels. It is easy to see that the POVMs $\{ M_1, \ldots, M_n \}$ and $\{ N_1, \ldots, N_m \}$ are compatible if and only if $\Phi_M$ and $\Phi_N$ are compatible as channels.

The result above also holds in more general settings. In particular, a similar result holds for measure-and-prepare channels if the choice of state preparations are \emph{distinguishable}. However, if the preparations are chosen in a general way, then this equivalence breaks down. To be precise, consider the channels defined as
\begin{equation}
\Psi'_M(X) = \sum_{i=1}^n \ip{M_i}{X} \, \rho_i
\quad \text{ and } \quad
\Psi'_N(X) = \sum_{j=1}^m \ip{N_j}{X} \, \sigma_j
\end{equation}
for some density operators $\rho_1, \ldots, \rho_n$ and $\sigma_1, \ldots, \sigma_m$. One can show that $\Psi'_M$ and $\Psi'_N$ are compatible if the POVMs $\{ M_1, \ldots, M_n \}$ and $\{ N_1, \ldots, N_m \}$ are compatible (but the converse may not be true).

%%%%%%%%%%%%%%%%%%%%%%%%%%%

\subsubsection*{Self-compatibility.}

Another notion of channel compatibility that we consider in this paper concerns the case when $\Phi_1=\Phi_2=\Phi $ for some fixed channel $\Phi$ (i.e., when the channels are the same). A channel $\Phi$ that is compatible with itself is said to be \emph{self-compatible}. Although most of this work is only concerned with the compatibility of pairs of channels, one may also consider---for some other positive integer $k>2$---whether some choice of $k$ channels $\Phi_1, \ldots, \Phi_k$ are compatible (see Section~\ref{section:background} for details and precise definitions). If $k$ copies of some fixed channel are compatible---i.e., $\Phi_1 = \cdots = \Phi_k=\Phi$ are all equal to some fixed channel $\Phi$---then we say that $\Phi$ is \emph{$k$-self-compatible}.

%%%%%%%%%%%%%%%%%%%%%%%%%%%%%%%%%%%%%%%%%%%%%%%%%%%%%% 

\subsection{Summary of results}

We examine the quantum channel marginal/compatibility problem using several different perspectives. Here we provide an overview of our results, which are stated in this section in an informal manner. Precise definitions and theorem statements can be found in subsequent sections.

%%%%%%%%%%%%%%%%%%%%%%%%%%%

\subsubsection*{The quantum channel marginal problem generalizes the state version.}
As is remarked above, the compatibility of channels is equivalent to the compatibility of their normalized Choi representations as quantum states. That is, the problem of determining the compatibility of channels can be reduced to solving the state marginal problem for a certain choice of states. We show that the quantum channel marginal problem also generalizes the state marginal problem in the following sense. Recently, a somewhat weaker version of this result was proved in \cite{HaapasaloKraftMiklinUola-marginalProblem}, where it was shown that the marginal problem for quantum states with invertible marginals is equivalent to the compatibility of channels. By using a different method to prove the result, we bypass the need for invertible marginal.

\begin{result}[Informal, see Theorem~\ref{thm:state--channel-equivalence} for a formal statement]
Every quantum state marginal problem is equivalent to the compatibility of a particular choice of quantum channels.
\end{result}

%%%%%%%%%%%%%%%%%%%%%%%%%%%

\subsubsection*{On the compatibility of measure-and-prepare.}

We first note that every measure-and-prepare (i.e., entanglement-breaking) channel is self-compa\-{}tible. Indeed, for measure-and-prepare channel $\Phi$ there exists a POVM $\{ M_1, \ldots, M_n \}$ and a collection of density matrices $\rho_1, \ldots, \rho_n$ such that
\begin{equation}
\Phi(X) = \sum_{i=1}^m \ip{X}{M_i} \, \rho_i,
\end{equation}
for all $X$. Now consider the channel
\begin{equation}
\label{intro:EBcomp}
\Psi(X) = \sum_{i=1}^m \ip{X}{M_i} \, \rho_i \otimes \rho_i.
\end{equation}
This channel clearly compatibilizes two copies of $\Phi$. (Moreover, one can easily modify $\Psi$ above such that it compatibilizes $k$ copies of $\Phi$. Hence every measure-and-prepare channel is also $k$-self-compatible for all $k$.)

The task of determining whether two \emph{distinct} measure-and-prepare channels $\Phi_1$ and $\Phi_2$ are compatible, however, is not so straightforward. From the discussion in the previous paragraph, if $\Phi_1$ and $\Phi_2$ are expressed as measure-and-prepare channels such that the prepared states are distinct computational basis states, it can be seen that the notion of channel compatibility is equivalent to that of measurement compatibility. However, this is not the case for all measure-and-prepare channels.  For instance, there may be multiple ways to express the measurements and/or preparations for a particular measure-and-prepare channel. One might guess that any channel which compatibilizes two measure-and-prepare channels must also be measure-and-prepare and can only exist if all three sets of measurements satisfy some conditions. However, we show that this is not the case, as described below.

\begin{result}[See Example~\ref{ex:no-ent-breaking-compatibilizer} for details]
There exists a pair of compatible measure-and-prepare channels with no measure-and-prepare compatibilizer.
\end{result}

The two channels were found using semidefinite programming formulations of channel compatibility and the notion of positive partial transpose (PPT). The channels and proof (which extracts away the SDP formalism) are presented in Section~\ref{sec:no-ent-breaking-compatibilizer}.

%%%%%%%%%%%%%%%%%%%%%%%%%%% 

\subsubsection*{Jordan products of quantum channels.}
The \emph{Jordan product} of two channels, $\Phi_1$ from system $\calX$ to system $\calY_1$ and $\Phi_2$ from $\calX$ to $\calY_2$, is the linear map $\Phi_1 \odot \Phi_2$ from the system $\calX$ to the system $\calY_1 \otimes \calY_2$ whose Choi representation is given by
\begin{equation}
\label{JordanChoi}
J(\Phi_1 \odot \Phi_2) = \sum_{i,j,k,\ell= 1}^{\dim(\calX)} ( E_{i,j} \odot E_{k,\ell} ) \otimes \Phi_1 ( E_{i,j} ) \otimes \Phi_2 ( E_{k,\ell} ),
\end{equation}
where $E_{i,j}$ is the matrix whose $(i,j)$-entry is 1 and has zeros elsewhere,\footnote{One can also think of $E_{i,j}$ as $| i \rangle \! \langle j |$ in Dirac notation.} and $\odot$ on the right-hand side denotes the Jordan product of matrices as defined in Equation~\eqref{jmat}. It is straightforward to see that the map $\Phi = \Phi_1 \odot \Phi_2$ satisfies
\begin{equation}
\label{eq:phi1-and-phi2-are-marginals-of-phi}
\Phi_1(X) = \Tr_{\calY_2} (\Phi(X))
\quad \text{ and } \quad
\Phi_2(X) = \Tr_{\calY_1} ( \Phi(X))
\end{equation}
for every choice of $X$. The map $\Phi_1 \odot \Phi_2$ might not be completely positive, as the corresponding Choi representation in~\eqref{JordanChoi} might not be positive semidefinite. However, if $\Phi_1 \odot \Phi_2$ is completely positive---and thus a channel, since it is trace-preserving by \eqref{eq:phi1-and-phi2-are-marginals-of-phi}---then this linear map is a compatibilizing channel for the channels $\Phi_1$ and $\Phi_2$. The condition that $\Phi_1\odot\Phi_2$ be completely positive is therefore a sufficient condition for the channels $\Phi_1$ and $\Phi_2$ to be compatible.

It is known that for projection-valued measures (PVMs)---that is POVMs where every operator is a projection---the Jordan product provides a necessary and sufficient condition for the compatibility of a PVM with any POVM \cite{HeinosaariReitznerStano-compat}. It is therefore natural to consider whether the Jordan product of channels similarly provides necessary and sufficient conditions for a PVM-measurement channel to be compatible with another arbitrary channel. Namely, for a PVM $\{ \Pi_1, \ldots, \Pi_m \}$, let $\Delta_{\Pi}$ be the corresponding measurement channel defined as
\begin{equation}
\Delta_{\Pi}(X) = \sum_{i=1}^m \Tr(\Pi_i X) E_{i,i}
\end{equation}
for every choice of $X$. One may ask whether complete positivity of the Jordan product $\Delta_{\Pi} \odot \Phi$ is always equivalent to the compatibility of $\Delta_\Pi$ and $\Phi$, for any other choice of channel $\Phi$. This is indeed the case, as described below.
\begin{result}[Informal, see Theorem~\ref{thm:jordan-PVM-NaS} for a formal statement]
$\Delta_{\Pi}$ is compatible with $\Phi$ if and only if $\Delta_{\Pi} \odot \Phi$ is completely positive.
\end{result}

To tackle more general cases, we describe how to relax the sufficient condition that the Jordan product be completely positive. Note that the Choi representation of the Jordan product as provided in \eqref{JordanChoi} can be expressed as
\begin{equation}
\label{intro:JP}
J(\Phi_1 \odot \Phi_2):= (\I_{\L(\calX)} \otimes \Phi_1 \otimes \Phi_2)(A_{\JP}) ,
\end{equation}
where $A_{\JP}$ is the Choi representation of the Jordan product of two identity channels (which we discuss in detail in Section~\ref{sec:jordanproduct}). By replacing $A_{\JP}$ in~\eqref{intro:JP} with another matrix $A \in \Herm(\calX \otimes \calX_1 \otimes \calX_2)$ satisfying
\begin{equation}
\label{eq:conditions-for-A-to-define-gen-Jordan-product}
\Tr_{\calX_1}(A) = \Tr_{\calX_2}(A) = \Tr_{\calX_1}(A_{\JP}) = \Tr_{\calX_2}(A_{\JP}) ,
\end{equation}
one obtains another linear map---which we denote by $\Phi_1 \odot_A \Phi_2$. Each such matrix $A$ provides another potential compatibilizing channel, so long as the corresponding map $\Phi_1 \odot_A \Phi_2$ is completely positive. If there exists at least one choice of Hermitian matrix $A$ satisfying \eqref{eq:conditions-for-A-to-define-gen-Jordan-product} such that the map $\Phi_1 \odot_A \Phi_2$ is completely positive, we say that the channels are \emph{Jordan compatible}. What is surprising is that this sufficient condition is also necessary in most cases, as described in the following two results.

\begin{result}[Informal, see Theorem~\ref{thm:jordan-gen-inverseIFF} for a formal statement]
If the channels $\Phi_1$ and $\Phi_2$ are invertible as linear maps, then they are compatible if and only if they are Jordan compatible. (Note that the inverses do not have to be quantum channels themselves.)
\end{result}

The requirement that both channels are invertible is not too restrictive, as indicated by the following result.

\begin{result}[Informal, see Theorem~\ref{thm:geometry-fullMeasure} for a formal statement]
{The set of Jordan-compatible pairs of channels has full measure as a subset of all compatible pairs.}
\end{result}

The Jordan product for quantum channels possesses many nice properties which are discussed in detail in Section~\ref{sec:jordanproduct}.

%%%%%%%%%%%%%%%%%%%%%%%%%%% 

\subsubsection*{Semidefinite programs for channel compatibility.}

The task of determining whether two channels $\Phi_1$ and $\Phi_2$ are compatible can be formulated as the following semidefinite programming feasibility problem:
\begin{equation}
\label{eq:channelcompatibility-sdp}
\begin{split}
\text{find:}\quad & X\in\Pos(\calX\otimes\calY_1\otimes\calY_2)\\
\text{satisfying:}\quad & \Tr_{\calY_2}(X) = J(\Phi_1)\\
& \Tr_{\calY_1}(X) = J(\Phi_2).
\end{split}
\end{equation}
This formulation can be found by using the Choi representations of each channel and their compatibilizer (where $X$ is the Choi representation of the desired compatibilizer).

The Jordan-compatibility of two channels can be similarly determined via the following semi\-{}definite programming feasibility problem:
\begin{equation}
\label{eq:jordan-product-sdp}
\begin{split}
\text{find:}\quad
& A\in\Herm(\calX\otimes\calX_1\otimes\calX_2)\\
\text{satisfying:}\quad
& \Tr_{\calX_1}(A) = \Tr_{\calX_1}(A_{\JP}) \\
& \Tr_{\calX_2}(A) = \Tr_{\calX_1}(A_{\JP}) \\
& (\I_{\L(\calX)} \otimes \Phi_1 \otimes \Phi_2)(A) \in\Pos(\calX\otimes\calY_1\otimes\calY_2) ,
\end{split}
\end{equation}
where $A_{\JP}$ is the matrix as determined in the discussion around~\eqref{intro:JP}.

The formulation in~\eqref{eq:channelcompatibility-sdp} has the advantage of being linear in the Choi representations of the channels---one could therefore keep them as variables and impose affine constraints on the channels. As a concrete example, one may ask whether there exist two qubit-to-qubit channels that are compatible and both unital. (We know that two identity channels do not satisfy these two conditions, but an identity channel and a completely depolarizing channel does.)

We use duality theory to show a few theorems of the alternative for the cases of compatibility and Jordan compatibility. For an example, we present one of two versions for compatibility, below.
\begin{result}[Informal, see Theorem~\ref{TotAV1} for a formal statement]
$\Phi_1$ and $\Phi_2$ are compatible
if and only if there does \emph{not} exist $Z_1$ and $Z_2$ such that
\begin{equation}
\Tr_{\calY_2}^*(Z_1) + \Tr_{\calY_1}^*(Z_2) \geq 0
\quad \text{ and } \quad \ip{Z_1}{J(\Phi_1)} + \ip{Z_2}{J(\Phi_2)} < 0.
\end{equation}
\end{result}
(Note that we define formally what the adjoint of the partial trace is later, but for now it can simply be viewed as a linear map.)

%%%%%%%%%%%%%%%%%%%%%%%%%%% 

\subsubsection*{Numerically testing qubit-to-qubit channels.}

We test various notions of compatibility for certain classes of qubit-to-qubit channels using the semidefinite programming formulations shown above. The family of channels that we consider are the \emph{partially dephasing-depolarizing} channels defined as
\begin{equation}
\Xi_{p,q} = (1-p-q)\I_{\L(\calX)}+p\Delta + q\Omega
\end{equation}
for parameters $p,q\in[0,1]$. Here $\I_{\L(\calX)}$ is the identity channel, $\Delta$ is the completely dephasing channel, and $\Omega$ is the completely depolarizing channel, which are defined by the equations
\begin{equation}
\label{qubitchannels}
\I_{\L(\calX)}(X) = X,\qquad
\Delta(X) = \sum_{i=1}^{\op{dim}(\calX)}\Tr(E_{i,i}X)E_{i,i},\qquad
\text{and}\qquad
\Omega(X) = \frac{\Tr(X)}{\op{dim}(\calX)}\I_{\calX}
\end{equation}
holding for all operators $X$, where $\I_\calX$ is the identity matrix.

In Section \ref{sec:qubitchanels}, we investigate the values $(p,q) \in [0,1] \times [0,1]$ for which the channel $\Xi_{p,q}$ is $k$-self-compatible for $k\in\{ 1, \ldots, 10 \} \cup \{ \infty \}$. (Recall that the condition that $\Xi_{p,q}$ is measure-and-prepare is equivalent to the condition that it is $k$-self\-compatible for all $k$.)

We also examine the values of $(p,q) \in [0,1] \times [0,1]$ for which $\Xi_{p,q} \odot \Xi_{p,q}$ is completely positive. This region turns out to be marginally smaller than the region where $\Xi_{p,q}$ is self-compatible, thus reinforcing the need for our generalization of the Jordan product. In fact, the channel $\Xi_{p,q}$ is invertible when $p, q > 0$ satisfies $p + q < 1$, so self-compatibility for this channel can be examined using (generalized) Jordan products for almost all values of $p$ and $q$.

Finally, we depict the region of values $(q_0, q_1) \in [0,1] \times [0,1]$ for which the pairs of channels $(\Xi_{0,q_0}, \Xi_{0,q_1}) = (\Omega_{q_0}, \Omega_{q_1})$ are compatible and when the Jordan product $\Omega_{q_0} \odot \Omega_{q_1}$ is completely positive. It turns out that the values of $(q_0, q_1)$ for which $\Omega_{q_0} \odot \Omega_{q_1}$ is completely positive contains some nontrivial instances while simultaneously missing other trivial instances of compatible pairs. 
This illustrates that the (standard) Jordan product provides an interesting sufficient condition for compatibility.

%%%%%%%%%%%%%%%%%%%%%%%%%%% 

\subsection{Paper organization}

The paper is organized as follows. In the next section, we introduce the notation and background for the work in this paper. In particular, we discuss quantum measurements and measure-and-prepare (i.e., entanglement-breaking) channels in Subsection~\ref{Sect22} and the compatibility of states, measurements, and channels in Subsection~\ref{Sect23}. We then study the equivalence between the quantum state marginal problem and the quantum channel marginal problem in Section~\ref{Sect3}. In Section~\ref{Sect4}, we study the compatibility of measure-and-prepare channels and show that sometimes they can be compatible without having a measure-and-prepare compatibilizer. In Section~\ref{sec:jordanproduct}, we introduce the Jordan product of quantum channels and study its connections to compatibility in depth. Section~\ref{section:geometry} continues our analysis of Jordan compatibility by studying the geometry of pairs of compatible and Jordan-compatible channels. In particular, we show that almost all compatible pairs of channels are Jordan compatible. We then study the question of determining compatibility through the lens of semidefinite programming and its duality theory in Section~\ref{SectSDP}. In Section~\ref{SectNumerical}, we study compatibility of pairs of qubit channels using numerical methods for solving such semidefinite programs. Lastly, we conclude and discuss open problems in Section~\ref{SectConclusions}.

%%%%%%%%%%%%%%%%%%%%%%%%%%% 
%%%%%%%%%%%%%%%%%%%%%%%%%%% 

\section{Notation and background}
\label{section:background}

In this section we summarize the terminology and notation used in this paper and review some of the basic known facts and results concerning compatibility in quantum theory. Note that we mostly use the same notation as in~\cite{Watrous-QI}.

%%%%%%%%%%%%%%%%%%%%%%%%%%% 

\subsection{Notation}

In this paper we work with (finite-dimensional) complex Euclidean spaces, which we may assume to be of the form $\complex^n$ for some positive integer $n$. We reserve the notation $\calX, \calX_1, \ldots, \calX_n$, $\calY, \calY_1, \ldots, \calY_n$ for complex Euclidean spaces. We use the following notation for frequently used sets of linear operators and linear maps of operators.
\begin{itemize}
\item $\L(\calX,\calY)$ is the space of linear operators from $\calX$ to $\calY$, and we write $\L(\calX)$ when $\calX=\calY$.
\item $\U(\calX,\calY)$ is the set of isometries from $\calX$ to $\calY$, while $\U(\calX)$ is the set of unitaries acting on $\calX$.
\item $\Herm(\calX)$ is the set of Hermitian operators acting on $\calX$.
\item $\Pos(\calX)$ is the set of Hermitian, positive semidefinite operators acting on $\calX$.
\item $\D(\calX)$ is the set of density matrices acting on $\calX$.
\item $\T(\calX, \calY)$ is the set of linear maps from $\L(\calX)$ to $\L(\calY)$, and we write $\T(\calX)$ when $\calX=\calY$.
\item $\C(\calX, \calY)$ is the set of quantum channels from $\calX$ to $\calY$, and we write $\C(\calX)$ when $\calX=\calY$.
\item $\Sep(\calX:\calY)$ is the set of separable operators acting on $\calX \otimes \calY$.
\item $\PPT(\calX:\calY)$ is the set of PPT operators acting on $\calX \otimes \calY$. (An operator $X\in\L(\calX\otimes\calY)$ is said to be PPT (positive partial transpose) if both $X$ and its partial transpose $X^{\t_\calX}$ are positive semidefinite.)
\end{itemize}

Every complex Euclidean space $\calX = \complex^n$ comes equipped with an inner product defined as $\ip{x}{y} = \sum_{i=1}^n\overline{x_i}y_i$ for every $x,y\in\calX$. For an operator $A \in \L(\calX,\calY)$, its adjoint $A^*\in\L(\calY,\calX)$ is the unique operator that satisfies $\ip{Ax}{y}=\ip{x}{A^*y}$ for every $x\in\calX$ and $y\in\calY$. The space $\L(\calX,\calY)$ is itself a complex Euclidean space with Hilbert--Schmidt inner product given by
\begin{equation}
\ip{A}{B}=\Tr(A^*B)
\end{equation}
for every $A,B\in\L(\calX,\calY)$. A self-adjoint operator $\Pi\in\L(\calX)$ is a projection operator if it satisfies $\Pi^2=\Pi$. For a subspace $\mathcal{V} \subseteq \calX$, we denote the projection operator onto the subspace $\mathcal{V}$ as the operator $\Pi_\mathcal{V}$. For an operator $X \in \L(\calX)$ we denote the image of $X$ as the subspace
\begin{equation}
\op{im}(X) = \{ Xu : u \in \calX \} \subseteq \calX.
\end{equation}
We use the notation $A \geq B$ to mean that the operator $A - B$ is positive semidefinite. The identity operator (i.e., identity matrix) on $\calX$ is denoted $\I_\calX$.

A quantum channel is a completely positive and trace-preserving linear map $\Phi\in\T(\calX,\calY)$. The identity channel on $\calX$ is denoted $\I_{\L(\calX)}$. Quantum channels are often presented in terms of their \emph{Choi representations}. For a linear map $\Phi\in\T(\calX,\calY)$, the Choi representation of $\Phi$ is the operator $J(\Phi) \in \L(\calX\otimes\calY)$ defined as
\begin{equation}
J(\Phi) = \sum_{i, j = 1}^{\dim(\calX)} E_{i,j} \otimes \Phi(E_{i,j}).
\end{equation}
The Choi representation satisfies the following properties:
\begin{itemize}
\item $\Phi$ is completely positive if and only if $J(\Phi) \geq 0$.
\item $\Phi$ is trace preserving if and only if $\Tr_\calY( J(\Phi)) = \I_\calX$, where $\Tr_\calY$ denotes the partial trace over~$\calY$.
\item $\Phi$ is entanglement-breaking if and only if $J(\Phi) \in \Sep(\calX : \calY)$.
\end{itemize}

One recovers the action of the linear map $\Phi \in \C(\calX, \calY)$ on a matrix $X \in \L(\calX)$ from its Choi representation by the equation
\begin{equation}
\Phi(X) = \Tr_{\calX}\bigl((X^\t\otimes\I_{\calY})J(\Phi)\bigr),
\end{equation}
where $X^\t$ denotes the transpose of an operator $X$. Note that the Choi representation of a linear map $\Phi\in\T(\calX,\calY)$ can also be defined via the equation
\begin{equation}
J(\Phi) = (\I_{\L(\calX)}\otimes\Phi)(J(\I_{\L(\calX)})).
\end{equation}
For a linear map $\Phi \in \T(\calX, \calY)$, its adjoint $\Phi^* \in \T(\calY, \calX)$ is the unique linear map that satisfies $\ip{\Phi(X)}{Y} = \ip{X}{\Phi^*(Y)}$ for every $X\in\L(\calX)$ and $Y\in\L(\calY)$. Note that the adjoint of a partial trace map is equivalent to tensoring with an identity operator. For example, the adjoint of the map $\Tr_{\calY} \in \T(\calX \otimes \calY, \calX)$ is given as
\begin{equation}
\Tr_{\calY}^*(X) = X \otimes \I_{\calY}
\end{equation}
for every $X \in \L(\calX)$. \\

%%%%%%%%%%%%%%%%%%%%%%%%%% 

\subsection{Background: Quantum measurements and measure-and-prepare channels}
\label{Sect22}

A \emph{measurement} in quantum theory is a procedure that assigns probabilities to quantum states. Every measurement on a complex Euclidean space $\calX$ is described by a \emph{positive-operator valued measure} (POVM)---a collection of positive semidefinite operators $\{M_i\,:\, i\in\Gamma\}\subset\Pos(\calX)$ for some finite set $\Gamma$ of \emph{measurement outcomes} that satisfies
\begin{equation}
\sum_{i\in\Gamma} M_i = \I_\calX.
\end{equation}
If the state of the system is described by a given density operator $\rho \in \D(\calX)$, the probability of obtaining a particular outcome $i\in\Gamma$ when measuring the system is equal to $\ip{M_i}{ \rho}$. Without loss of generality, the finite set of measurement outcomes may be assumed to be of the form $\Gamma=\{ 1, \dots, m \}$ for some positive integer $m$.

A \emph{projective-valued measure} (PVM) is a POVM of the form $\{\Pi_i\,:\, i\in\Gamma\}$ such that the operator $\Pi_i$ is a projection operator for each outcome $i\in\Gamma$. One necessarily has that $\Pi_i\Pi_j=0$ for distinct outcomes $i, j \in \Gamma$.

Every POVM $\{M_1,\dots,M_m\}$ has an associated \emph{measurement channel}, which is the linear map $\Phi_M\in\C(\calX,\calY)$ defined as
\begin{equation}
\label{eq:measchannel}
\Phi_M(X) = \sum_{i = 1}^m \ip{M_i}{X} E_{i,i}
\end{equation}
for every $X\in\L(\calX)$, where one defines $\calY=\complex^m$. It is easy to check that $\Phi_M$ is completely positive and trace preserving. Measure-and-prepare channels are a generalization of the idea above and are considered often in this work.

\begin{definition}[Measure-and-prepare channel]
A \emph{measure-and-prepare} (or \emph{entanglement-breaking}) \linebreak channel $\Phi\in \C(\calX, \calY)$ is a channel for which there is a choice of POVM $\{M_1,\dots,M_m\}\subset\Pos(\calX)$ and density matrices $\rho_1,\dots,\rho_m \in \D(\calY)$ such that $\Phi$ may be expressed as
\begin{equation}
\Phi(X) = \sum_{i=1}^m \ip{M_i}{X}\, \rho_i
\end{equation}
for every $X\in\L(\calX)$. The channel $\Phi$ is said to be \emph{generated by} the POVM $\{M_1,\dots,M_m\}$ and the density matrices $\rho_1, \ldots, \rho_m$.
\end{definition}

Note that there is not typically a unique choice of POVM or collection of density operators that generate a given measure-and-prepare channel. Recall that the condition that a channel ${\Phi\in\C(\calX,\calY)}$ be measure-and-prepare is equivalent to the condition that $J(\Phi)\in\op{Sep}(\calX:\calY)$ (i.e., its Choi representation is a separable operator).

%%%%%%%%%%%%%%%%%%%%%%%%%% 

\subsection{Compatibility of states, measurements, and channels}
\label{Sect23}

We now consider the problem of compatibility for quantum states.

\begin{definition}[Quantum state marginal problem]
Two states $\rho_1 \in \D(\calX \otimes \calY_1)$ and $\rho_2 \in \D(\calX \otimes \calY_2)$ are said to be \emph{compatible} if there exists another state $\rho \in \D(\calX \otimes \calY_1 \otimes \calY_2)$ satisfying
\begin{equation}
\Tr_{\calY_2}(\rho) = \rho_1
\quad \text{ and } \quad
\Tr_{\calY_1}(\rho) = \rho_2.
\end{equation}
The state $\rho$ is said to be the \emph{joint state} for $\rho_1$ and $\rho_2$. Determining whether two overlapping states are compatible is known as the \emph{quantum state marginal problem}.
\end{definition}

Suppose that $\{M_1,\dots,M_m\}\subset\Pos(\calX)$ and $\{N_1,\dots,N_n\}\subset\Pos(\calX)$ are POVMs for some complex Euclidean space $\calX$. One can ask whether it is possible to obtain the statistics of both measurements as course-grainings from a single measurement. If so, the measurements are said to be compatible.

\begin{definition}[Measurement compatibility]
Two POVMS $\{M_1,\dots,M_m\}\subset\Pos(\calX)$ and \linebreak $\{N_1,\dots,N_n\}\subset\Pos(\calX)$ are said to be \emph{compatible} if there exists another POVM of the form
\begin{equation}
\{P_{i,j}\,:\, 1\leq i\leq m,\, 1\leq j\leq n\}\subset \Pos(\calX)
\end{equation}
satisfying
\begin{align}
M_i = \sum_{j=1}^n P_{i,j}\qquad \text{and}\qquad N_j = \sum_{i=1}^m P_{i,j}
\end{align}
for each choice of indices $i\in\{1,\dots,m\}$ and $j\in\{1,\dots,n\}$. A POVM $P$ satisfying these conditions is said to be a \emph{compatibilizer} (or \emph{compatibilizing measurement}) for $M$ and $N$.
\end{definition}

Compatibility of POVMs has been investigated before (see \cite{HeinosaariMiyaderaZiman-review} for a review). Analogous to the compatibility of POVMs, one may consider whether two channels may obtained as the marginals of another larger channel. This notion of compatibility is defined below.

\begin{definition}[Channel compatibility] \label{def:intro-channelComapatibility}
Two channels $\Phi_1 \in \C(\calX, \calY_1)$ and $\Phi_2 \in \C(\calX, \calY_2)$ are \emph{compatible} if there exists a channel $\Phi \in \C(\calX, \calY_1 \otimes \calY_2)$ satisfying
\begin{align}
\Phi_1(X) = \Tr_{\calY_2} ( \Phi(X) ) \qquad\text{and}\qquad \Phi_2(X) = \Tr_{\calY_1} ( \Phi(X) )
\end{align}
for every $X\in\L(\calX)$. The channel $\Phi$ is said to be a \emph{compatibilizer} (or \emph{compatibilizing channel}) for $\Phi_1$ and $\Phi_2$.
\end{definition}

Compatibility of channels has also been investigated recently (see, e.g.,~\cite{HeinosaariMiyadera-channels, Plavala-channels, Kuramochi-channels, Kuramochi-ETBchannels, Haapasalo-compat}). Note that two compatible channels do not necessarily possess a unique compatibilizer.\footnote{For a simple example of a pair of channels that do not possess a unique compatibilizer, consider the completely dephasing channel defined by the equation $\Omega(X) = \Tr(X)\I_{\calX}/\op{dim}(\calX)$. This channel is trivially self-compatible, but there are infinitely many choices compatibilizing channel for two copies of $\Omega$. Indeed, any channel $\Phi\in\T(\calX,\calX_1\otimes\calX_2)$ having the form $\Phi(X)=\Tr(X)\rho$, where $\calX=\calX_1=\calX_2$ and $\rho\in\D(\calX_1\otimes\calX_2)$ is a maximally entangled state, will compatibilize two copies of $\Omega$. This follows from the fact that $\Tr_{\calX_1}(\rho)=\Tr_{\calX_2}(\rho) = \I_{\calX}/\op{dim}(\calX)$ holds for every choice of maximally entangled state $\rho\in\D(\calX_1\otimes\calX_2)$.} A channel $\Phi\in\C(\calX,\calY)$ is said to be \emph{self-compatible} if two copies of $\Phi$ are compatible.

It is natural to generalize the notion of compatibility to more than two channels. In particular, for $k \geq 2$, a collection of channels $\Phi_1\in\C(\calX,\calY_1),\dots,\Phi_k\in\C(\calX,\calY_k)$ are said to be compatible if there exists a compatibilizing channel $\Phi\in\C(\calX,\calY_1\otimes\cdots\otimes\calY_k)$ such that
\begin{equation}
\Tr_{\calY_1 \otimes \cdots \otimes \calY_k \setminus \calY_a}(\Phi(X)) = \Phi_a(X)
\end{equation}
holds for every $X \in \L(\calX)$ and each index $a \in \{ 1, \ldots, k \}$. A channel $\Phi \in \C(\calX, \calY)$ is said to be \emph{$k$-self-compatible} if $k$ copies of the channel $\Phi$ are compatible. It is known that a channel is measure and prepare if and only if it is $k$-self-compatible for every positive integer $k \in \mathbb{N}$.

%%%%%%%%%%%%%%%%%%%%%%%%%%%%%%%%%%%%%%%%%%%%%%%%%%%%%%

\section{Equivalence of the quantum state and quantum channel versions of the marginal problem}
\label{Sect3}

It is straightforward to see that the marginal problem for quantum states generalizes the problem of determining compatibility for quantum channels. Indeed, two channels are compatible precisely when their corresponding normalized Choi representations are compatible as states \cite{Plavala-channels, HaapasaloKraftMiklinUola-marginalProblem}. To see this, suppose a pair of channels $\Phi_1 \in \C(\calX, \calY_1)$ and $\Phi_2 \in \C(\calX, \calY_2)$ are compatible with some choice of compatibilizing channel $\Phi \in \C(\calX, \calY_1 \otimes \calY_2)$. Consider now the condition that
\begin{equation}
\label{Choistates}
\Tr_{\calY_2} \circ \, \Phi = \Phi_1.
\end{equation}
The left- and right-hand side of this equality must be equal as maps, so their Choi representations must also coincide. Let $\rho$, $\rho_1$, and $\rho_2$ be the states defined by normalizing versions of the Choi representations of $\Phi$, $\Phi_1$, and $\Phi_2$, respectively, i.e.,
\begin{equation}
\rho = \frac{1}{\op{dim}(\calX)}J(\Phi), \qquad \rho_1 = \frac{1}{\op{dim}(\calX_1)}J(\Phi_1), \qquad \text{and}\qquad \rho_2 = \frac{1}{\op{dim}(\calX)}J(\Phi_2).
\end{equation}
Then the equality in~\eqref{Choistates} is equivalent to the condition tat
\begin{equation}
\Tr_{\calY_2}(\rho) = \rho_1.
\end{equation}
Similarly, the condition that $\Tr_{\calY_1} \circ \, \Phi = \Phi_2$ is equivalent to the condition tat
\begin{equation}
\Tr_{\calY_1}(\rho) = \rho_2.
\end{equation}
Hence the question of compatibility for a pair of quantum channels can be reduced to the corresponding marginal problem for the corresponding states defined as the normalized Choi representations.

It is also the case that the state marginal problem can be reduced to the channel compatibility problem. In particular, every quantum state marginal problem can be reformulated as the problem of determining compatibility of some pair of channels. This has already been established in \cite{HaapasaloKraftMiklinUola-marginalProblem} in the case when the marginals of the states have full rank. In this section, we prove this reduction holds even in the more general case when the marginals do not necessarily have full rank.

Let $\rho_1 \in \D(\calX \otimes \calY_1)$ and $\rho_2 \in \D(\calX \otimes \calY_2)$ be a pair of compatible quantum states. It must be the case that
\begin{equation}
\Tr_{\calY_1}(\rho_1) = \Tr_{\calY_1 \otimes \calY_2}(\rho) = \Tr_{\calY_2}(\rho_2),
\end{equation}
where $\rho$ is some choice of joint state. Therefore we may assume without loss of generality that
\begin{equation}
\Tr_{\calY_1}(\rho_1) = \Tr_{\calY_2}(\rho_2),
\end{equation}
as otherwise the pair of states would be clearly not compatible. If it holds further that these marginals are both equal to $\Tr_{\calY_1}(\rho_1) = \Tr_{\calY_2}(\rho_2) = \I_\calX/\op{dim}(\calX)$ (i.e., the maximally mixed state), one may simply view $\rho_1$ and $\rho_2$ as the normalizations of the Choi representations for some channels $\Phi_1 \in \C(\calX, \calY_1)$ and $\Phi_2 \in \C(\calX, \calY_2)$. In this case, the marginal problem for $\rho_1$ and $\rho_2$ is trivially equivalent to the problem of determining the compatibility of the channels $\Phi_1$ and $\Phi_2$.

The more general case, when the marginal state $\Tr_{\calY_1}(\rho_1) = \Tr_{\calY_2}(\rho_2)$ is not proportional to the identity operator, is considered in Theorem \ref{thm:state--channel-equivalence}. We first prove the following useful lemma, which will assist us in considering the case when the marginal does not have full rank.

\begin{lemma}
\label{lemma:quantumMarginal-subspace}
Let $A \in \Pos(\calX \otimes \calY)$ be a positive operator and let $\Pi := \Pi_{\op{im}(\Tr_{\calY}(A))}$ be the projection operator onto $\op{im}(\Tr_{\calY}(A)) \subseteq \calX$. It holds that
\begin{equation}
A = (\Pi \otimes \I_{\calY}) \, A \, (\Pi\otimes \I_{\calY}).
\end{equation}
\end{lemma}

\begin{proof}
Note that
\begin{equation}
\ip{A}{(\I_{\calX}-\Pi)\otimes\I_{\calY}} = \ip{\Tr_\calY(A)}{\I_{\calX}-\Pi} = 0.
\end{equation}
As both $A$ and $(\I_{\calX}-\Pi)\otimes\I_{\calY}$ are positive operators, it follows that $A((\I_{\calX}-\Pi)\otimes\I_{\calY}) = 0 $ and $((\I_{\calX}-\Pi)\otimes\I_{\calY})A = 0$. The desired result follows.
\end{proof}

We are now ready to show that the marginal problem for any pair of states, $\rho_1\in\D(\calX\otimes\calY_1)$ and $\rho_2\in\D(\calX\otimes\calY_2)$, can be reduced to the compatibility problem for a particular pair of channels, $\Phi_1\in\C(\calX,\calY_1)$ and $\Phi_2\in\C(\calX,\calY_2)$ and.

\begin{theorem}
\label{thm:state--channel-equivalence}
Let $\rho_1\in\D(\calX\otimes\calY_1)$ and $\rho_2\in\D(\calX\otimes\calY_2)$ be states and suppose there is a state $\sigma\in\D(\calX)$ satisfying
\begin{equation}
\sigma = \Tr_{\calY_1}(\rho_1) = \Tr_{\calY_2}(\rho_2).
\end{equation}
Let $\Phi_1\in\T(\calX,\calY_1)$ and $\Phi_2\in\T(\calX,\calY_2)$ be the linear maps whose Choi representations may be expressed as
\begin{equation}
\label{eq:eq1}
\begin{split}
J(\Phi_1) &= (\sigma^{-\frac{1}{2}}\otimes\I_{\calY_1})\rho_1(\sigma^{-\frac{1}{2}}\otimes\I_{\calY_1}) + \frac{1}{\op{dim}(\calY_1)}(\I_\calX-\Pi_{\op{im}(\sigma)})\otimes\I_{\calY_1}\\\text{and}\qquad J(\Phi_2) &= (\sigma^{-\frac{1}{2}}\otimes\I_{\calY_2})\rho_2(\sigma^{-\frac{1}{2}}\otimes\I_{\calY_2}) + \frac{1}{\op{dim}(\calY_2)}(\I_\calX-\Pi_{\op{im}(\sigma)})\otimes\I_{\calY_2}
\end{split}
\end{equation}
(where we interpret $\sigma^{-\frac{1}{2}}$ as the Moore--Penrose pseudoinverse of $\sigma^{\frac{1}{2}}$ if $\sigma$ is not invertible). The maps $\Phi_1$ and $\Phi_2$ are channels. Moreover, the operators $\rho_1$ and $\rho_2$ are compatible as states if and only if $\Phi_1$ and $\Phi_2$ are compatible as channels.
\end{theorem}
\begin{proof}
Note that $\Phi_1$ and $\Phi_2$ are completely positive, as each of the terms in the sums on the right-hand sides of the equalities in \eqref{eq:eq1} are positive semidefinite operators. These maps are also trace preserving, as
\begin{equation}
\Tr_{\calY_1}(J(\Phi_1)) = \sigma^{-\frac{1}{2}}\sigma\sigma^{-\frac{1}{2}} + \I_{\calX}-\Pi_{\op{im}(\sigma)} = \I_{\calX}
\end{equation}
(where we note that $\Pi_{\op{im}(\sigma)} = \sigma^{-\frac{1}{2}}\sigma\sigma^{-\frac{1}{2}}$) and similarly $\Tr_{\calY_2}(J(\Phi_2))=\I_{\calX}$. We may conclude that the maps $\Phi_1$ and $\Phi_2$ are channels, as claimed.

Now suppose that $\rho_1$ and $\rho_2$ are compatible and let $\rho\in\D(\calX\otimes\calY_1\otimes\calY_2)$ be a choice of joint state for $\rho_1$ and $\rho_2$. Let $\Phi\in\T(\calX,\calY_1\otimes\calY_2)$ be the linear map whose Choi representation may be expressed as
\begin{equation}
\label{eq:choi}
\begin{split}
J(\Phi) &= (\sigma^{-\frac{1}{2}}\otimes\I_{\calY_1}\otimes\I_{\calY_2})\rho(\sigma^{-\frac{1}{2}}\otimes\I_{\calY_1}\otimes\I_{\calY_2}) \\
&+ \dfrac{1}{\op{dim}(\calY_1)\op{dim}(\calY_2)}(\I_\calX-\Pi_{\op{im}(\sigma)})\otimes\I_{\calY_1}\otimes\I_{\calY_2}.
\end{split}
\end{equation}
Note that $\Phi$ is completely positive, as its Choi representation is expressed in \eqref{eq:choi} as the sum of positive semidefinite operators. It is straightforward to verify that the map defined in this manner satisfies $\Tr_{\calY_2}(J(\Phi)) = J(\Phi_1)$ and $\Tr_{\calY_1}(J(\Phi)) = J(\Phi_2)$. Thus $\Phi$ is a compatibilizer for $\Phi_1$ and $\Phi_2$.

Suppose instead that $\Phi_1$ and $\Phi_2$ are compatible as channels and let $\Phi\in\C(\calX,\calY_1\otimes\calY_2)$ be a choice of compatibilizing channel. Define an operator $\rho\in\L(\calX\otimes\calY_1\otimes\calY_2)$ as
\begin{equation}
\rho = (\sigma^{\frac{1}{2}}\otimes\I_{\calY_1}\otimes\I_{\calY_2})J(\Phi)(\sigma^{\frac{1}{2}}\otimes\I_{\calY_1}\otimes\I_{\calY_2}).
\end{equation}
It is evident that $\rho$ is positive semidefinite. Recalling the definition of $\Phi_1$ from \eqref{eq:eq1}, one has that
\begin{align*}
\Tr_{\calY_2}(\rho) &= (\sigma^{\frac{1}{2}}\otimes\I_{\calY_1})\Tr_{\calY_2}(J(\Phi))(\sigma^{\frac{1}{2}}\otimes\I_{\calY_1})\nonumber\\
&= (\sigma^{\frac{1}{2}}\otimes\I_{\calY_1})J(\Phi_1)(\sigma^{\frac{1}{2}}\otimes\I_{\calY_1})\nonumber\\
& = (\Pi_{\op{im}(\sigma)}\otimes\I_{\calY_1})\rho_1(\Pi_{\op{im}(\sigma)}\otimes\I_{\calY_1})\\
& = \rho_1 ,
\end{align*}
where the equality in the third line follows from the facts that
\begin{equation}
\sigma^{\frac{1}{2}}\sigma^{-\frac{1}{2}} = \sigma^{-\frac{1}{2}}\sigma^{\frac{1}{2}} = \Pi_{\op{im}(\sigma)}
\qquad\text{and}\qquad
\sigma^{\frac{1}{2}}(\I_{\calX} - \Pi_{\op{im}(\sigma)})\sigma^{\frac{1}{2}} = 0,
\end{equation}
and equality in the fourth line follows from Lemma \ref{lemma:quantumMarginal-subspace}. The proof that that $\Tr_{\calY_1}(\rho) = \rho_2$ is analogous. It follows that $\rho$ is a joint state for $\rho_1$ and $\rho_2$, which completes the proof.
\end{proof}

%%%%%%%%%%%%%%%%%%%%%%%%%%%%%%%%%%%%%%%%%%%%%%%%%%%%%%

\section{Compatibility of measure-and-prepare channels}
\label{Sect4}

In this section, we consider some facts about the compatibility of measure-and-prepare channels (and also for the special case of measurement channels). We review some known conditions for compatibility of certain classes of measure-and-prepare channels. More importantly, we also consider the following question: If a pair of measure-and-prepare channels $\Phi_1\in\C(\calX, \calY_1)$ and $\Phi_2\in\C(\calX, \calY_2)$ are compatible, do they they necessarily possess a measure-and-prepare compatibilizer? We show that this is not the case by providing a counterexample.

%%%%%%%%%%%%%%%%%%%%%%%%%%%

\subsection{Compatibility with measurement channels}

Here we state a necessary and sufficient condition for a channel to be compatible with a fixed choice of measurement channel. Recall that a measurement channel is a special instance of a measure-and-prepare channel, where the measurement outcomes are recorded in the computation basis. The equivalence in the following proposition was shown in the more general context of generalized probability theories in \cite[Lemma 1]{HeinosaariLeppajarviPlavala-noFreeInformation}.

\begin{proposition}[\cite{HeinosaariLeppajarviPlavala-noFreeInformation}]\label{prop:meascompat}
Let $\{M_1,\dots,M_m\}\subset\Pos(\calX)$ be a POVM and let $\Phi_M$ be the measurement channel $\Phi_M$ generated by $M$,
\begin{equation}
\Phi_M(X) = \sum_{i=1}^m\ip{M_i}{X} E_{i,i}.
\end{equation}
Let $\Phi \in \C(\calX,\calY)$ be a channel. The following are equivalent.
\begin{enumerate}
\item The channels $\Phi$ and $\Phi_M$ are compatible.
\item There exist completely positive maps $\Phi_1,\dots,\Phi_m\in\T(\calX,\calY)$ satisfying $\Phi=\sum_{i=1}^m\Phi_i$ such that, for each $i\in\{1,\dots,m\}$, one has $\Tr(\Phi_i(X)) = \ip{M_i}{X}$ for every $X\in\L(\calX)$.
\end{enumerate}
\end{proposition}

We remark that the condition that $\Tr(\Phi_i(X)) = \ip{M_i}{X}$ holds for every $X\in\L(\calX)$ is equivalent to the condition that $\Tr_{\calY}(J(\Phi_i))=M_i^\t$. Indeed, one has that
\begin{equation}
\Tr(\Phi_i(X)) = \ip{J(\Phi_i)}{X^\t\otimes \I_{\calY}} = \ip{\Tr_\calY(J(\Phi_i))^\t}{X}.
\end{equation}
for every $X\in\L(\calX)$, and the desired equivalence follows.

%%%%%%%%%%%%%%%%

\subsection{Conditions for the compatibility of measure-and-prepare channels}

Here we review some known results regarding the compatibility of measure-and-prepare channels, which are stated rigorously in Proposition \ref{prop:compatible measurement channels}. The first statement is that two measure-and-prepare channels are compatible as channels whenever their underlying POVMs are compatible as measurements. A stronger result holds if the collections of preparation states defining the measure-and-prepare channels are assumed to be distinguishable. Allow us to take a moment to recall what this means. A collection of density operators $\{\rho_1, \ldots, \rho_m \} \subset \D(\calY)$ is said to be (\emph{perfectly}) \emph{distinguishable} if it holds that $\rho_i \rho_j=0$ for every pair of indices $i, j \in \{ 1, \ldots, n \}$ satisfying $i \neq j$. In particular, this means that there exists a POVM that can perfectly distinguish between these states and, moreover, that such a POVM can be given by the projection operators, i.e, it is given as $\{ \Pi_{\op{im}(\rho_1)},\dots,\Pi_{\op{im}(\rho_m)} \}$, where the projections must be pairwise orthogonal.

\begin{proposition}
\label{PropositionEB}\label{prop:compatible measurement channels}
Let $\Phi_1 \in \C(\calX, \calY_1)$ and $\Phi_2 \in \C(\calX, \calY_2)$ be measure-and-prepare channels having the form
\begin{equation}
\Phi_1(X) = \sum_{i=1}^m \ip{M_i}{X} \rho_i
\qquad\text{and}\qquad
\Phi_1(X) = \sum_{j=1}^n \ip{N_j}{X} \sigma_j
\end{equation}
for some choice of POVMs $\{M_1,\dots,M_m\}\subset\Pos(\calX)$ and $\{N_1,\dots,N_n\}\subset \Pos(\calX)$, and some collections of density matrices $\rho_1, \ldots, \rho_m \in \D(\calY_1)$ and $\sigma_1, \ldots, \sigma_n \in \D(\calY_2)$. The following statements hold:
\begin{enumerate}[label = (\arabic*)]
\item\label{item:compatible-measurements-channels-1} If $M$ and $N$ are compatible as POVMS then $\Phi_1$ and $\Phi_2$ are compatible as channels.
\item\label{item:compatible-measurements-channels-2} If the collections of density matrices $\rho_1, \ldots, \rho_m$ and $\sigma_1, \ldots, \sigma_n$ are each distinguishable then the channels $\Phi_1$ and $\Phi_2$ are compatible if and only if $M$ and $N$ are compatible as POVMs.
\end{enumerate}
\end{proposition}

\begin{proof}
To prove statement \ref{item:compatible-measurements-channels-1}, suppose that $M$ and $N$ are compatible and let
\begin{equation}
\{P_{i,j}\,:\, 1\leq i\leq m,\, 1\leq j\leq n\}\subset\Pos(\calX)
\end{equation}
be a choice of POVM that compatibilizes $M$ and $N$. Define a channel $\Phi \in \C(\calX,\calY_1 \otimes \calY_2)$ as
\begin{equation}
\Phi(X) = \sum_{i=1}^n \sum_{j=1}^m \ip{P_{i,j}}{X} \rho_i \otimes \sigma_j
\end{equation}
for every $X\in\L(\calX)$. Taking partial traces, it may be verified that $\Phi$ compatibilizes $\Phi_1$ and $\Phi_2$.

One direction of statement \ref{item:compatible-measurements-channels-2} follows from statement \ref{item:compatible-measurements-channels-1}. To prove the reverse implication, assume that the channels $\Phi_1$ and $\Phi_2$ are compatible and let $\Phi\in\C(\calX,\calY_1\otimes\calY_2)$ be a compatibilizing channel. For each choice of indices $i\in\{1,\dots,m\}$ and $j\in\{1,\dots,n\}$, one may define the operator
\begin{equation}
P_{i,j} = \Phi^*(\Pi_{\op{im}(\rho_i)}\otimes\Pi_{\op{im}(\sigma_j)}) \in \Pos(\calX).
\end{equation}
We may assume without loss of generality that $\sum_{i=1}^{m} \Pi_{\op{im}(\rho_i)} = \I_{\calY_1}$. Indeed, the projection operators $\Pi_{\op{im}(\rho_1)},\dots,\Pi_{\op{im}(\rho_m)}$ must be pairwise orthogonal, as it has been assumed that the density matrices $\rho_1, \ldots, \rho_m $ are distinguishable, and if these projections did not sum to the identity, the operator $\Pi = \I_{\calX}-\sum_{i=1}^m\Pi_{\op{im}(\rho_i)}$ would be a nontrivial projection operator. One could then define the operator $M_{m+1}=0$ and the state $\rho_{m+1} = \Pi/\Tr(\Pi)$ such that the POVM $\{M_1,\dots,M_{m+1}\}$ and collection of density matrices $\rho_1,\dots,\rho_{m+1}$ satisfy the desired property and generate the same measurement channel, which justifies our assumption. Similarly, it may be assumed without loss of generality that $\sum_{j=1}^{n} \Pi_{\op{im}(\sigma_j)} = \I_{\calY_2}$.

For each index $i\in\{1,\dots,m\}$, one has that
\begin{equation}
\sum_{j=1}^n \ip{P_{i,j}}{X}
= \ip{\Pi_{\op{im}(\rho_i)}\otimes\I_{\calY_2}}{\Phi(X)} = \ip{\Pi_{\op{im}(\rho_i)}}{\Phi_1(X)} = \ip{M_i}{X}
\end{equation}
every $X\in\L(\calX)$, and thus $\sum_{j=1}^n P_{i,j} = M_i$. Similarly, for each index $j\in\{1,\dots,n\}$,
\begin{equation}
\sum_{i=1}^m \ip{P_{i,j}}{X}
= \ip{\I_{\calY_1}\otimes\Pi_{\op{im}(\sigma_j)})}{\Phi(X)} = \ip{\Pi_{\op{im}(\sigma_j)}}{\Phi_2(X)} = \ip{N_j}{X}
\end{equation}
holds for every $X\in\L(\calX)$, and thus $\sum_{i=1}^m P_{i,j} = N_j$. From these two conditions, it can be verified that $P$ is a compatibilizing POVM for $M$ and $N$, so the POVMs are compatible.
\end{proof}

%%%%%%%%%%%%%%%%%%%%%%%%%%% 

\subsection{Compatible measure-and-prepare channels do not necessarily have a measure-and-prepare compatibilizer}
\label{sec:no-ent-breaking-compatibilizer}

Every measure-and-prepare channel $\Phi \in \C(\calX, \calY)$ is necessarily self-compatible and, moreover, there is a measure-and-prepare channel $\Psi\in\C(\calX,\calY\otimes\calY)$ that compatibilizes two copies of $\Phi$. Indeed, suppose $\{M_1,\dots,M_m\}\subset \Pos(\calX)$ is a POVM and $\rho_1, \ldots, \rho_m \in \D(\calY)$ are density matrices such that $\Phi$ may be expressed as
\begin{equation}
\Phi(X) = \sum_{i=1}^m \ip{M_i}{X}\,\rho_i.
\end{equation}
One may define a measure-and-prepare compatibilizer $\Psi\in\C(\calX,\calY\otimes\calY)$ as
\begin{equation}
\Psi(X) =\sum_{i=1}^m \ip{M_i}{X}\,\rho_i \otimes\rho_i.
\end{equation}
We now consider the question of whether two compatible measure-and-prepare channels necessarily possess a measure-and-prepare compatibilizer. It turns out that the answer is no, which we demonstrate with the following example.

\begin{example}
\label{ex:no-ent-breaking-compatibilizer}
Let $\calX=\calY_1= \calY_2=\complex^2$ and let $\Phi_1\in\C(\calX,\calY_1)$ and $\Phi_2\in\C(\calX,\calY_2)$ be the channels whose Choi representations $J(\Phi_1)\in\L(\calX\otimes\calY_1)$ and $J(\Phi_2)\in\L(\calX\otimes\calY_2)$ are given by the $4\times 4$ matrices
\begin{equation}
\label{eq:choiphi1phi2}
J(\Phi_1) = \left(
\begin{array}{cccc}
\frac{3}{4} & \cdot & \cdot & \frac{1}{4} \\
\cdot & \frac{1}{4} & \cdot & \cdot \\
\cdot & \cdot & \frac{1}{4} & \cdot \\
\frac{1}{4} & \cdot & \cdot & \frac{3}{4} \\
\end{array}
\right)
\qquad\text{and}
\qquad
J(\Phi_2) = \left(
\begin{array}{cccc}
\frac{5}{8} & \cdot & \cdot & \frac{3}{8} \\
\cdot & \frac{3}{8} & \frac{1}{8} & \cdot \\
\cdot & \frac{1}{8} & \frac{3}{8} & \cdot \\
\frac{3}{8} & \cdot & \cdot & \frac{5}{8} \\
\end{array}
\right),
\end{equation}
where $\cdot$ indicates an entry that is $0$. One may observe that both $J(\Phi_1)$ and $J(\Phi_2)$ as well as their partial transposes (where $Z^{\tX}$ denotes the partial transpose of an operator $Z \in \L(\calX\otimes\calY)$),
\begin{equation}
J(\Phi_1)^{\tX} = \left(
\begin{array}{cccc}
\frac{3}{4} & \cdot & \cdot & \cdot \\
\cdot & \frac{1}{4} & \frac{1}{4} & \cdot \\
\cdot & \frac{1}{4} & \frac{1}{4} & \cdot \\
\cdot & \cdot & \cdot & \frac{3}{4} \\
\end{array}
\right)
\qquad\text{and}
\qquad
J(\Phi_2)^{\tX} = \left(
\begin{array}{cccc}
\frac{5}{8} & \cdot & \cdot & \frac{1}{8} \\
\cdot & \frac{3}{8} & \frac{3}{8} & \cdot \\
\cdot & \frac{3}{8} & \frac{3}{8} & \cdot \\
\frac{1}{8} & \cdot & \cdot & \frac{5}{8} \\
\end{array}
\right)
\end{equation}
are positive semidefinite matrices. The operators $J(\Phi_1)$ and $J(\Phi_2)$ are therefore positive under partial transpose (PPT). It follows that $\Phi_1$ and $\Phi_2$ are necessarily measure-and-prepare channels as their Choi representations are separable as operators.\footnote{Recall that a channel is measure-and-prepare channel if and only if its Choi representation is a separable operator. Moreover, for bipartite operators where each local dimension is equal to $2$, an operator is separable if and only if it is PPT~\cite{HorodeckiHorodeckiHorodecki-PPT}.} Consider now the map $\Phi\in\T(\calX,\calY_1\otimes\calY_2)$ whose Choi representation $J(\Phi)\in\L(\calX\otimes\calY_1\otimes\calY_2)$ is given by the $8\times 8$ matrix
\begin{equation}
J(\Phi) = \left(
\begin{array}{cccccccc}
\frac{1}{2} & \cdot & \cdot & \cdot & \cdot & \frac{3}{16} & \frac{1}{8} & \cdot \\
\cdot & \frac{1}{4} & \cdot & \cdot & \frac{1}{16} & \cdot & \cdot & \frac{1}{8} \\
\cdot & \cdot & \frac{1}{8} & \cdot & \cdot & \cdot & \cdot & \frac{3}{16} \\
\cdot & \cdot & \cdot & \frac{1}{8} & \cdot & \cdot & \frac{1}{16} & \cdot \\
\cdot & \frac{1}{16} & \cdot & \cdot & \frac{1}{8} & \cdot & \cdot & \cdot \\
\frac{3}{16} & \cdot & \cdot & \cdot & \cdot & \frac{1}{8} & \cdot & \cdot \\
\frac{1}{8} & \cdot & \cdot & \frac{1}{16} & \cdot & \cdot & \frac{1}{4} & \cdot \\
\cdot & \frac{1}{8} & \frac{3}{16} & \cdot & \cdot & \cdot & \cdot & \frac{1}{2} \\
\end{array}
\right).
\end{equation}
It is easy to verify that this map satisfies $ \Tr_{\calY_2}(J(\Phi)) = J(\Phi_1) $ and $ \Tr_{\calY_1}(J(\Phi)) = J(\Phi_2) $. Moreover, it may be verified that $J(\Phi)$ is positive definite, as the eigenvalues of $J(\Phi)$ are the four possible values of
\begin{equation}
\frac{4\pm\sqrt{3}+\sqrt{10\pm4\sqrt{3}}}{16},
\end{equation}
each having multiplicity $2$. Therefore $\Phi$ is a channel that compatibilizes $\Phi_1$ and $\Phi_2$, and thus the measure-and-prepare channels $\Phi_1$ and $\Phi_2$ are compatible.

It remains to show that there does not exist a measure-and-prepare compatibilizer for the channels $\Phi_1$ and $\Phi_2$. Toward a contradiction, suppose there exists some choice of measure-and-prepare channel $\Psi \in \C(\calX,\calY_1\otimes\calY_2)$ that compatibilizes the channels $\Phi_1$ and $\Phi_2$. This channel must satisfy\footnote{Here we make use of the fact that every separable operator must be PPT.}
\begin{equation}
\label{PPTcondition}
\Tr_{\calY_2}(J(\Psi)) = J(\Phi_1) , \qquad \Tr_{\calY_1}(J(\Psi)) = J(\Phi_2), \qquad \text{and} \qquad J(\Psi)^{\tX} \geq 0.
\end{equation}
For any choice of Hermitian operators $Z_1 \in \L(\calX \otimes \calY_1)$ and $Z_2 \in \L(\calX \otimes \calY_2)$ satisfying
\begin{equation}
\label{eq:Z1Y2Z2Y1}
\Tr_{\calY_2}^*(Z_1)+ \Tr_{\calY_1}^*(Z_2) \geq 0,
\end{equation}
it must be the case that
\begin{align}
\ip{J(\Phi_1)^{\tX}}{Z_1} + \ip{J(\Phi_2)^{\tX}}{Z_2}
& = \ip{\Tr_{\calY_2}(J(\Psi)^{\tX})}{Z_1} + \ip{\Tr_{\calY_1}(J(\Psi)^{\tX})}{Z_2} \\
& = \ip{J(\Psi)^{\tX}}{\Tr_{\calY_2}^*(Z_1)+ \Tr_{\calY_1}^*(Z_2)}\\
&\geq0,
\end{align}
where the inequality follows from the fact that the expression in the second line is an inner product of two operators that are positive semidefinite by assumption, and $\Psi$ is any channel that satisfies the conditions in \eqref{PPTcondition}. {(Here we make use of the adjoint maps $\Tr_{\calY_1}^*$ and $\Tr_{\calY_2}^*$ of the corresponding partial traces. To preserve the correct ordering of the spaces in the tensor product, these are defined here as the maps that satisfy
\begin{equation}
\label{eq:Tr*definition}
\Tr_{\calY_1}^*(X\otimes Y_2) = X\otimes\I_{\calY_1}\otimes Y_2
\qquad\text{and}\qquad
\Tr_{\calY_2}^*(X\otimes Y_1) = X\otimes Y_1\otimes\I_{\calY_2}
\end{equation}
for every $X\in\L(\calX)$, $Y_1\in\L(\calY_1)$, and $Y_2\in\L(\calY_2)$.)}

To prove that there does \emph{not} exist a measure-and-prepare compatibilizer, it therefore suffices to find a choice of operators $Z_1\in\L(\calX\otimes\calY_1)$ and $Z_2\in\L(\calX\otimes\calY_2)$ satisfying
\begin{equation}
\label{earlier adjoints}
\Tr_{\calY_2}^*(Z_1)+ \Tr_{\calY_1}^*(Z_2) \geq 0
\qquad\text{and}\qquad
\ip{J(\Phi_1)^{\tX}}{Z_1} + \ip{J(\Phi_2)^{\tX}}{Z_2} < 0 .
\end{equation}
Toward this goal, let $Z_1$ and $Z_2$ be the operators defined by the matrices
\begin{equation}
Z_1 = \left(
\begin{array}{cccc}
-2 & 0 & 0 & 2 \\
0 & 48 & -38 & 0 \\ 0 & -38 & 48 & 0 \\ 2 & 0 & 0 & -2 \\
\end{array}
\right)
\qquad\text{and}\qquad
Z_2 =\left(
\begin{array}{cccc}
3 & 0 & 0 & -4 \\ 0 & 40 & -47 & 0 \\ 0 & -47 & 40 & 0 \\ -4 & 0 & 0 & 3 \\
\end{array}
\right).
\end{equation}
For this choice of operators, one has that
\begin{equation}
\ip{J(\Phi_1)^{\tX}}{Z_1} + \ip{J(\Phi_2)^{\tX}}{Z_2}
= - \frac{1}{2} < 0.
\end{equation}
However, the matrix
\begin{equation}
\Tr_{\calY_2}^*(Z_1) + \Tr_{\calY_1}^*(Z_2)
=
\left(
\begin{array}{cccccccc}
1 & 0 & 0 & 0 & 0 & -4 & 2 & 0 \\ 0 & 38 & 0 & 0 & -47 & 0 & 0 & 2 \\ 0 & 0 & 51 & 0 & -38 & 0 & 0 & -4 \\ 0 & 0 & 0 & 88 & 0 & -38 & -47 & 0 \\ 0 & -47 & -38 & 0 & 88 & 0 & 0 & 0 \\ -4 & 0 & 0 & -38 & 0 & 51 & 0 & 0 \\ 2 & 0 & 0 & -47 & 0 & 0 & 38 & 0 \\ 0 & 2 & -4 & 0 & 0 & 0 & 0 & 1 \\
\end{array}
\right)
\end{equation}
may be verified to be positive semidefinite. This contradicts the assumption that there exists a measure-and-prepare compatibilizer for $\Phi_1$ and $\Phi_2$, and thus no such compatiblizer exists.
\end{example}

\begin{remark}
The matrices above where found numerically using the semidefinite programming formulations discussed later in this work. However, the proof above does not rely on semidefinite programming and can be presented without this formalism.
\end{remark}

%%%%%%%%%%%%%%%%%%%%%%%%%%%%%%%%%%%%%%%%%%%%%%%%%%%%%%%%%%%%%%

\section{The Jordan product of channels}
\label{sec:jordanproduct}

For any pair of operators $A,B\in\L(\calX)$, their \emph{Jordan product} is defined as
\begin{equation}
A \odot B := \dfrac{1}{2} (AB + BA).
\end{equation}
The Jordan product provides a useful sufficient condition for verifying that two POVMs are compatible. Let $\{M_1,\dots,M_m\}\subset\Pos(\calX)$ and $\{N_1,\dots,N_n\}\subset\Pos(\calX)$ be two POVMs. If it holds that
\begin{equation}
\label{eq:jordan product positive}
M_i\odot N_j\geq0
\end{equation}
for each pair of indices $i$ and $j$, then $M$ and $N$ are compatible. Indeed, one may define a choice of compatibilizing POVM as $P_{i,j}=M_i\odot N_j$ \cite{Heinosaari-coexistence}. If at least one of the POVMs $M$ or $N$ is also a \emph{projection-valued measure} (PVM), then the condition in \eqref{eq:jordan product positive} is also a necessary condition for these POVMs to be compatible \cite{HeinosaariReitznerStano-compat}.

In this section, we generalize the notion of Jordan products for operators to linear maps using the Choi representation. We also investigate properties of Jordan products of quantum channels.

\subsection{Definition and properties}
\label{sec:JordanProdDef}

Here we define the notion of Jordan products of linear maps and discuss a few of its properties.

\begin{definition}[Jordan product of linear maps] \label{def:jordan-prod}
Let $\Phi_1 \in \T(\calX, \calY_1)$ and $\Phi_2 \in \T(\calX, \calY_2)$ be linear maps. The \emph{Jordan product} of $\Phi_1$ and $\Phi_2$ is the linear map
\begin{equation}
\Phi_1 \odot \Phi_2 \in \T(\calX, \calY_1 \otimes \calY_2)
\end{equation}
whose Choi representation $J(\Phi_1 \odot \Phi_2)\in\L(\calX\otimes\calY_1\otimes\calY_2)$ is the operator
\begin{equation}
J(\Phi_1 \odot \Phi_2) := \sum_{i,j,k,\ell = 1}^{\dim(\calX)} ( E_{i,j} \odot E_{k,\ell} ) \otimes \Phi_1 ( E_{i,j} ) \otimes \Phi_2 ( E_{k,\ell} ). \label{eq:jordan-def}
\end{equation}
\end{definition}

Many useful properties of the Jordan product for linear maps follow immediately from the definition, some of which are outlined as follows. The binary operation $\odot$ is bilinear---that is, for all choices of linear maps $\Phi_1, \Phi_1' \in \T(\calX, \calY_1)$ and $\Phi_2, \Phi_2' \in \T(\calX, \calY_2)$ and scalars $\alpha, \beta \in \complex$, one has that
\begin{equation}
\begin{split}
(\alpha \Phi_1 + \beta \Phi_1') \odot \Phi_2 &= \alpha\, \Phi_1 \odot \Phi_2 + \beta\, \Phi_1' \odot \Phi_2 \\
\text{and }\quad\Phi_1 \odot (\alpha \Phi_2 + \beta \Phi_2') &= \alpha\, \Phi_1 \odot \Phi_2 + \beta\, \Phi_1 \odot \Phi_2' .
\end{split}
\end{equation}
For any further choices of linear maps $\Psi_1 \in \T(\calY_1, \calY_1')$ and $\Psi_2 \in \T(\calY_2, \calY_2')$, one has that
\begin{equation}
(\Psi_1\otimes\Psi_2)\circ(\Phi_1\odot\Phi_2) = (\Psi_1\circ\Phi_1)\odot(\Psi_2\circ\Phi_2).
\end{equation}
It is also evident that
\begin{equation}
(\Phi_2\odot\Phi_1)(X) = \op{Swap}_{\calY_1,\calY_2}\bigl((\Phi_1\odot\Phi_2)(X)\bigr)
\end{equation}
holds for every $X\in\L(\calX)$, where $\op{Swap}_{\calY_1,\calY_2} \in \T(\calY_1 \otimes \calY_2, \calY_2 \otimes \calY_1)$ is the linear map defined by the equation
\begin{equation}
\op{Swap}_{\calY_1,\calY_2}(Y_1\otimes Y_2)=Y_2\otimes Y_1
\end{equation}
for every $Y_1\in\L(\calY_1)$ and $Y_2\in\L(\calY_2)$.

It is useful to note a few alternative ways of representing this Jordan product. In particular, the Choi representation of the Jordan product of $\Phi_1 \in \T(\calX, \calY_1)$ and $\Phi_2 \in \T(\calX, \calY_2)$ can be expressed as
\begin{equation}
\label{eq:ChoiJordanProduct}
J(\Phi_1 \odot \Phi_2) = \frac{1}{2} \sum_{i,j=1}^{\dim(\calX)} E_{i,j} \otimes \left( \sum_{k=1}^{\dim(\calX)}\Phi_1(E_{i,k})\otimes\Phi_2(E_{k,j}) + \Phi_1(E_{k,j})\otimes\Phi_2(E_{i,k}) \right).
\end{equation}
It is then straightforward to check that
\begin{equation}
(\Phi_1 \odot \Phi_2)(X) = \frac{1}{2}(\Phi_1 \otimes \Phi_2)\bigl(W(X\otimes\I_{\calX}+\I_\calX\otimes X)\bigr)
\end{equation}
for every $X\in\L(\calX)$, where $W\in\U(\calX\otimes\calX)$ is the unitary operator defined by the equation
\begin{equation}
W(x\otimes y) = y \otimes x \; \text{ for every } \; x, y \in \calX.
\end{equation}
Indeed, making use of the expression in \eqref{eq:ChoiJordanProduct}, one has
\begin{align}
(\Phi_1 \odot \Phi_2)(X) &= \Tr_{\calX} \bigl(J(\Phi_1 \odot \Phi_2)(X^\t\otimes \I_{\calY_1 \otimes \calY_2}) \bigr) \\
&= \dfrac{1}{2} (\Phi_1 \otimes \Phi_2) \left( \sum_{i,j,k = 1}^{\dim(\calX)} \ip{E_{i,j}}{X} ( E_{i,k} \otimes E_{k,j} + E_{k,j} \otimes E_{i,k} ) \right) \\
&= \dfrac{1}{2} (\Phi_1 \otimes \Phi_2) \left( \sum_{i,j,k = 1}^{\dim(\calX)} \ip{E_{i,j}}{X} W ( E_{k,k} \otimes E_{i,j} + E_{i,j} \otimes E_{k,k} ) \right) \\
&= \dfrac{1}{2} (\Phi_1 \otimes \Phi_2) ( W ( \I_{\calX} \otimes X + X \otimes \I_{\calX} ) )
\end{align}
where we have used the identity $X = \sum_{i,j=1}^{\dim(\calX)} \ip{E_{i,j}}{X} E_{i,j}$.

Other interesting properties of Jordan products arise when considering the Jordan product of channels. In particular, the Jordan product of trace-preserving maps is also trace preserving, as the following proposition shows.

\begin{proposition}
\label{prop:jordan-trace-preserving}
Suppose that $\Phi_1 \in \T(\calX, \calY_1)$ and $\Phi_2 \in \T(\calX, \calY_2)$ are trace-preserving maps. One has that
\begin{equation}
\label{eq:tracepreservingJordan}
\begin{split}
\Tr_{\calY_2} ( J(\Phi_1 \odot \Phi_2) ) = J(\Phi_1)
\quad \text{ and } \quad
\Tr_{\calY_1} ( J(\Phi_1 \odot \Phi_2) ) = J(\Phi_2).
\end{split}
\end{equation}
Moreover, $\Phi_1 \odot \Phi_2$ is trace preserving.
\end{proposition}

\begin{proof}
Using~\eqref{eq:ChoiJordanProduct}, we can write
\begin{align}
\Tr_{\calY_2} ( J(\Phi_1 \odot \Phi_2) ) &= \frac{1}{2} \sum_{i,j=1}^{\dim(\calX)} E_{i,j}\otimes \left(
\sum_{k=1}^{\dim(\calX)} \Tr(\Phi_2(E_{k,j}))\Phi_1 ( E_{i,k} ) + \Tr(\Phi_2(E_{i,k}))\Phi_1 ( E_{k,j} )
\right) \\
&=\frac{1}{2} \sum_{i,j=1}^{\dim(\calX)} E_{i,j}\otimes \bigl( \Phi_1 ( E_{i,j} ) + \Phi_1 ( E_{i,j} ) \bigr) \\
&= J(\Phi_1),
\end{align}
where we made use of the fact that $\Phi_2$ is trace preserving. The proof of the other equality in \eqref{eq:tracepreservingJordan} is analogous. Finally, note that
\begin{equation}
\Tr_{\calY_1\otimes\calY_2}(J(\Phi_1\odot\Phi_2)) = \Tr_{\calY_1}(J(\Phi_1)) = \I_\calX
\end{equation}
and thus $\Phi_1\odot\Phi_2$ is trace preserving.
\end{proof}

An immediate consequence of Proposition \ref{prop:jordan-trace-preserving} is that the Jordan product provides us with a useful criteria for determining if two channels are compatible. In particular, given two channels $\Phi_1 \in \C(\calX, \calY_1)$ and $\Phi_2 \in \C(\calX, \calY_2)$, one straightforward way to check if $\Phi_1$ and $\Phi_2$ are compatible is to check if the Jordan product map $\Phi_1 \odot \Phi_2$ is completely positive.

\begin{corollary}
\label{cor:jordan-suff-condition}
Suppose that $\Phi_1 \in \C(\calX, \calY_1)$ and $\Phi_2 \in \C(\calX, \calY_2)$ are such that $\Phi_1 \odot \Phi_2$ is completely positive. Then $\Phi_1 \odot \Phi_2$ compatibilizes $\Phi_1$ and $\Phi_2$.
\end{corollary}

As this is only a sufficient condition, other methods must be used to determine the compatibility of the channels whose Jordan product is not completely positive.

If the maps $\Phi_1$ and $\Phi_2$ are both measure-and-prepare channels, the Jordan product map $\Phi_1 \odot \Phi_2$ can be obtained simply by taking the Jordan products of the elements of the POVMs that generate $\Phi_1$ and $\Phi_2$, discussed in the proposition below. It is in this sense that the sufficient condition for the compatibility of measurements via the Jordan product of the POVM elements is a generalization of the similar sufficient condition via the Jordan product for channels.

\begin{proposition}
Let $\Phi_1 \in \C(\calX, \calY_1)$ and $\Phi_2 \in \C(\calX, \calY_2)$ be measure-and-prepare channels that are generated by the POVMs $\{M_1,\dots,M_m\}\subset \Pos(\calX)$ and $\{N_1,\dots,N_n\}\subset\Pos(\calX)$, and the collections of density matrices $\rho_1, \ldots, \rho_m \in \D(\calY_1)$ and $\sigma_1, \ldots, \sigma_n \in \D(\calY_2)$, respectively. It holds that
\begin{equation}
J( \Phi_1 \odot \Phi_2 ) = \sum_{i=1}^{m}\sum_{j=1}^n \bigl( M_i \odot N_j \bigr)^\t\otimes \rho_i \otimes \sigma_j . \label{eq:jordan-ETBprod}
\end{equation}
Moreover, if $M_i \odot N_j \geq 0$ for each index $i$ and $j$, then the map $\Phi_1 \odot \Phi_2$ is completely positive (and thus compatible).
\end{proposition}

\begin{proof}
This follows by direct calculation from the definition of the Jordan product of maps.
\end{proof}

There are a few further properties of Jordan products of channels that follow directly from the definition that we mention here. Given channels $\Phi_1, \Phi_1' \in \C(\calX, \calY_1)$, a channel $\Phi_2 \in \C(\calX, \calY_2)$, and a scalar $\lambda \in [0, 1]$, if it holds that $J(\Phi_1 \odot \Phi_2) \geq 0$ and $J(\Phi'_1 \odot \Phi_2) \geq 0$ then
\begin{equation}
J((\lambda \Phi_1 + (1-\lambda) \Phi'_1) \odot \Phi_2) \geq 0.
\end{equation}
Furthermore, for any other choice of channel $\Psi \in \C(\calY_1, \calY'_1)$, if it holds that $J(\Phi_1 \odot \Phi_2) \geq 0$ then $J( (\Psi \circ \Phi_1) \odot \Phi_2) \geq 0$ as well.

%%%%%%%%%%%%%%%%%%%%%%%%%%%%%%%%%%%%%%%%%%%%%%%%%%%%%%%%%%%%%%

\subsection{Necessary and sufficient conditions for the compatibility from the Jordan product}
Recall that a \emph{projection-valued measure} (PVM) is a collection $\Pi=\{\Pi_1,\dots,\Pi_m \}$ of projection operators that satisfy $\sum_{i=1}^m \Pi_i = \I_\calX$. Let $\Pi=\{\Pi_1,\dots,\Pi_m\}$ be a PVM and let $M=\{M_1,\dots,M_n\}$ be an arbitrary POVM. It is known that $\Pi$ and $M$ are compatible if and only if each of the Jordan product operators $\Pi_i \odot M_j$ is positive semidefinite \cite{HeinosaariReitznerStano-compat}. In this section, we discuss a similar result for channels.

For the PVM $\Pi=\{\Pi_1,\dots,\Pi_m \}$, consider the measurement channel $\Delta_\Pi\in\C(\calX,\calZ)$ defined as
\begin{equation}
\Delta_\Pi(X) = \sum_{i=1}^{m}\ip{\Pi_i}{X} E_{i,i}
\end{equation}
for every $X\in\L(\calX)$, where $\op{dim}(\calZ) = m$. Theorem \ref{thm:jordan-PVM-NaS} provides necessary and sufficient criteria for a given channel $\Phi\in\C(\calX,\calY)$ to be compatible with $\Delta_\Pi$ in terms of their Jordan product.

\begin{theorem}
\label{thm:jordan-PVM-NaS}
Let $\{ \Pi_1,\dots,\Pi_m \}$ be a PVM. Define the channels $\Xi_\Pi\in\C(\calX)$ and $\Delta_{\Pi} \in \C(\calX)$ as
\begin{equation}
\label{eq:XiPiDeltaPi}
\Xi_\Pi(X) = \sum_{i=1}^m \Pi_iX\Pi_i\qquad\text{and}\qquad\Delta_\Pi (X) = \sum_{i=1}^m \ip{\Pi_i}{X} E_{i,i}
\end{equation}
for every $X\in\L(\calX)$, and let $\Phi \in \C(\calX,\calY)$ be a channel. The following are equivalent.
\begin{enumerate}[label = (\arabic*)]
\item\label{item:jordan-PVM-Nas-2} The map $\Phi \odot \Delta_\Pi$ is completely positive.
\item\label{item:jordan-PVM-Nas-1} The channels $\Phi$ and $\Delta_\Pi$ are compatible.
\item\label{item:jordan-PVM-Nas-3} It holds that $\Phi = \Phi\circ\Xi_\Pi$.
\end{enumerate}
\end{theorem}
Note that the map $\Xi_\Pi$ as defined in Theorem \ref{thm:jordan-PVM-NaS} is called the \emph{pinching} channel corresponding to the PVM $\Pi$.

\begin{proof}
The implication \ref{item:jordan-PVM-Nas-2}$\implies$\ref{item:jordan-PVM-Nas-1} follows from Corollary \ref{cor:jordan-suff-condition}. To prove the implication \ref{item:jordan-PVM-Nas-1}$\implies$\ref{item:jordan-PVM-Nas-3}, suppose that the channels $\Phi$ and $\Delta_\Pi$ are compatible. By Proposition \ref{prop:meascompat}, there exist completely positive maps $\Phi_1,\dots,\Phi_m\in\T(\calX,\calY)$ such that $\Phi = \sum_{i=1}^m\Phi_i$ and satisfy $\Tr_{\calY}(J(\Phi_i)) = \Pi_i$ for each index $i\in\{1,\dots,m\}$. It follows from Lemma \ref{lemma:quantumMarginal-subspace}, for each index $i\in\{1,\dots,m\}$, that
\begin{equation}
J(\Phi_i) = (\Pi_i\otimes\I_\calY)J(\Phi_i)(\Pi_i\otimes\I_\calY)
\end{equation}
since $\Pi_{\op{im}(\Tr_{\calY}(J(\Phi_i)))}= \Pi_i$, and thus
\begin{equation}
\label{eq:PhiiPiiXPii}
\Phi_i(X)= \Phi_i(\Pi_iX\Pi_i)
\end{equation}
holds for every $X\in\L(\calX)$. Making use of the equality in \eqref{eq:PhiiPiiXPii}, and the fact that $\Pi_i\Pi_j =0$ whenever $i\neq j$, we have that
\begin{align}
(\Phi\circ\Xi_\Pi)(X) &= \sum_{i=1}^m\sum_{j=1}^m \Phi_i(\Pi_jX\Pi_j)
= \sum_{i=1}^m\sum_{j=1}^m \Phi_i(\Pi_i\Pi_jX\Pi_j\Pi_i)\\
&= \sum_{i=1}^m\Phi_i(\Pi_iX\Pi_i) = \sum_{i=1}^m \Phi_i(X) = \Phi(X)
\end{align}
for every $X\in\L(\calX)$, and thus $\Phi = \Phi\circ\Xi_\Pi$.

Finally, to prove the implication \ref{item:jordan-PVM-Nas-3}$\implies$\ref{item:jordan-PVM-Nas-2}, suppose that $\Phi = \Phi\circ\Xi_\Pi$. Note that the channels $\Xi_\Pi$ and $\Delta_\Pi$ are compatible. Indeed, it may be verified that the Jordan product of these channels can be expressed as
\begin{equation}
(\Xi_\Pi\odot\Delta_\Pi)(X) = \sum_{i=1}^m \Pi_iX\Pi_i\otimes E_{i,i}.
\end{equation}
This map is clearly completely positive and therefore compatibilizes $\Xi_\Pi$ and $\Delta_\Pi$. Hence the map
\begin{equation}
\Phi\circ\Delta_\Pi = (\Phi\circ\Xi_\Pi)\odot\Delta_\Pi = (\Phi\otimes\I_{\L(\calZ)})\circ(\Xi_\Pi\circ\Delta_\Pi)
\end{equation}
is completely positive, as it is the composition of completely positive maps. This completes the proof.
\end{proof}

If $\Pi = \{E_{1,1},\dots,E_{n,n}\}$ is the PVM consisting of the rank-one projections in the computation basis (where $n=\op{dim}(\calX)$), the channel $\Delta_\Pi=\Delta$ is simply the \emph{completely dephasing channel} on $\calX$,
\begin{equation}
\label{triangle}
\Delta(X) = \sum_{i=1}^{n}\ip{E_{i,i}}{X} E_{i,i}.
\end{equation}
The map $\Delta$ is the measurement channel corresponding to measuring in the computational basis. Applying the result of Theorem \ref{thm:jordan-PVM-NaS} in this case provides us with the following necessary and sufficient conditions for a channel to be compatible with the completely dephasing channel $\Delta$.

\begin{corollary}
\label{prop:jordan-Delta-NaS}
Let $\Phi \in \C(\calX,\calY)$ be a channel and let $\Delta\in\C(\calX)$ be the completely dephasing channel. The following are equivalent.
\begin{enumerate}[label = (\arabic*)]
\item\label{item:prop:jordan-Delta-NaS-1} The map $\Phi \odot \Delta$ is completely positive.
\item\label{item:prop:jordan-Delta-NaS-2} The channels $\Phi$ and $\Delta$ are compatible.
\item\label{item:prop:jordan-Delta-NaS-3} It holds that $\Phi$ is a measure-and-prepare channel generated by the computational basis PVM.
\end{enumerate}
\end{corollary}

\begin{proof}
Let $n=\op{dim}(\calX)$. For the PVM $\Pi = \{E_{1,1},\dots,E_{n,n}\}$, it is evident that $\Xi_\Pi = \Delta_\Pi= \Delta$, where $\Xi_\Pi$ and $\Delta_\Pi$ are the channels as defined in \eqref{eq:XiPiDeltaPi}. By Theorem~\ref{thm:jordan-PVM-NaS}, it will suffice to show that the condition in \ref{item:prop:jordan-Delta-NaS-3} is equivalent to the condition that $\Phi = \Phi\circ\Delta$. If $\Phi\circ\Delta = \Phi$ then
\begin{equation}
\label{eq:PhiDelta}
\Phi(X) = \Phi(\Delta(X)) = \sum_{i=1}^n \ip{E_{i,i}}{X}\Phi(E_{i,i}) = \sum_{i=1}^n \ip{E_{i,i}}{X}\rho_i,
\end{equation}
for every $X\in\L(\calX)$, where $\rho_1,\dots,\rho_n\in\D(\calY)$ are the density matrices defined as $\rho_i=\Phi(E_{i,i})$ for each $i$, and thus $\Phi$ satisfies condition \ref{item:prop:jordan-Delta-NaS-3}. On the other hand, if $\Phi$ satisfies condition \ref{item:prop:jordan-Delta-NaS-3} then there exist density matrices $\rho_1,\dots,\rho_n$ such that $\Phi$ may be expressed as
\begin{equation}
\Phi(X) = \sum_{i=1}^n\ip{E_{i,i}}{X}\rho_i,
\end{equation}
and thus
\begin{equation}
\Phi(\Delta(X)) = \sum_{i=1}^n\sum_{j=1}^n\ip{E_{i,i}}{\ip{E_{j,j}}{X}E_{j,j}}\rho_i = \sum_{i=1}^n\ip{E_{i,i}}{X}\rho_i = \Phi(X)
\end{equation}
holds for every $X\in\L(\calX)$.
\end{proof}

For the sake of completeness we also investigate other extreme cases.
\begin{lemma}
\label{lemma:jordan-unitary-NaS}
Let $\calX$ and let $\Phi_U \in \C(\calX)$ be a unitary channel, i.e., $\Phi_U(X) = UXU^*$ for some $U \in \U(\calX)$. Let $\Psi \in \C(\calX, \calY)$. Then the following are equivalent.
\begin{enumerate}[label = (\arabic*)]
\item\label{item:jordan-unitary-Nas-1} $\Phi_U$ and $\Psi$ are compatible.
\item\label{item:jordan-unitary-Nas-2} $\Psi$ is a constant channel, i.e., there is a fixed state $\rho \in \D(\calY)$ such that $\Psi(X) = \Tr(X) \rho$ for all $X \in \L(\calX)$.
\item\label{item:jordan-unitary-Nas-3} $J(\Phi_U \odot \Psi) \geq 0$.
\end{enumerate}
\end{lemma}
\begin{proof}
To prove the implication \ref{item:jordan-unitary-Nas-1}$\implies$\ref{item:jordan-unitary-Nas-2}, note that $J(\Phi_U)$ is a multiple of a maximally entangled state. Suppose $\Phi \in \C(\calX, \calX' \otimes \calY)$ is a compatibilizer of $\Phi_U$ and $\Psi$, where, for convenience, we define $\calX'=\calX$. It follows that $\Tr_{\calY}( J(\Phi) ) = J(\Phi_U)$. By the monogamy of entanglement, it holds that
\begin{equation}
J(\Phi) = J(\Phi_U) \otimes Y.
\end{equation}
for some fixed choice of positive operator $Y \in \Pos(\calY)$. Now,
\begin{equation}
J(\Psi) = \Tr_{\calX'}( J(\Phi) ) = \Tr_{\calX'}( J(\Phi_U) \otimes Y ) = \I_\calX \otimes Y.
\end{equation}
For all $X \in \L(\calX)$, one has that
\begin{equation}
\Psi(X) = \Tr_{\calX} ( (X^\t \otimes \I_{\calY}) J(\Psi) ) = \Tr_{\calX} ( (X^\t \otimes Y ) = \Tr(X) Y.
\end{equation}
The fact that $\Tr(Y) = 1$ (and thus $Y$ is a state) follows from the fact that $\Psi$ must be trace preserving.

To prove the implication \ref{item:jordan-unitary-Nas-2}$\implies$\ref{item:jordan-unitary-Nas-3}, let $\rho\in\D(\calY)$ be a state such that $\Psi(X) = \Tr(X) \rho$ for every $X \in \L(\calX)$. Note that
\begin{equation}
J(\Phi_U \odot \Psi) = J(\Phi_U) \otimes \rho,
\end{equation}
and thus the operator $J(\Phi_U \odot \Psi)$ is clearly positive semidefinite. The implication \mbox{\ref{item:jordan-unitary-Nas-3}$\implies$\ref{item:jordan-unitary-Nas-1}} follows from Corollary \ref{cor:jordan-suff-condition}.
\end{proof}

\begin{lemma}
\label{lemma:jordan-constant-NaS}
Let $\Phi_\rho \in \C(\calX, \calY)$ be a constant channel, i.e., $\Phi_\rho(X) = \Tr(X) \rho$ for all $X \in \L(X)$, for some fixed choice of state $\rho \in \D(\calY)$. Let $\Psi \in \C(\calX, \calZ)$. It holds that $J(\Phi_\rho \odot \Psi) \geq 0$ and the channels $\Phi_\rho$ and $\Psi$ are compatible.
\end{lemma}
\begin{proof}
As in the proof of Lemma \ref{lemma:jordan-unitary-NaS}, it is evident that
\begin{equation}
J(\Psi \odot \Phi_\rho) = J(\Psi) \otimes \rho,
\end{equation}
where we have changed the order of the channels for convenience.
\end{proof}

%%%%%%%%%%%%%%%%%%%%%%%%%%%%%%%%%%%%%%%%%%%%%%%%%%%%%%%%%%%%%%

\subsection{Generalized Jordan products}

One may ask whether the Jordan product of maps given by Definition \ref{def:jordan-prod} is the unique construction that satisfies the key properties outlined in Section \ref{sec:JordanProdDef}. We now show that this is not the case by generalizing the construction of the Jordan product of linear maps.

\begin{definition}[Generalized Jordan product of linear maps] \label{def:genJordProd}
Let $\calX=\calX_1=\calX_2$ be complex Euclidean spaces and let $A\in\Herm(\calX\otimes\calX_1\otimes\calX_2)$ be an operator satisfying
\begin{equation}
\label{eq:Tr1ATr2A}
\Tr_{\calX_1}(A) = \Tr_{\calX_2}(A) = \sum_{i,j=1}^{\dim(\calX)}E_{i,j}\otimes E_{i,j}.
\end{equation}
Let $\Phi_1 \in \T(\calX, \calY_1)$ and $\Phi_2 \in \T(\calX, \calY_2)$ be linear maps. The \emph{generalized Jordan product} of the maps $\Phi_1$ and $\Phi_2$ with respect to the operator~$A$ is the linear map
\begin{equation}
\Phi_1\odot_A\Phi_2 \in \T(\calX, \calY_1 \otimes \calY_2)
\end{equation}
whose Choi representation is the operator $J(\Phi_1\odot_A\Phi_2) \in \L(\calX \otimes \calY_1 \otimes \calY_2)$ given by
\begin{equation}
J(\Phi_1\odot_A\Phi_2) = \bigl(\I_{\L(\calX)}\otimes\Phi_1\otimes\Phi_2\bigr)(A).
\end{equation}
\end{definition}
Note that the requirement in \eqref{eq:Tr1ATr2A} can be rephrased as
\begin{equation}
\Tr_{\calX_1}(A) = \Tr_{\calX_2}(A) = J(\I_{\L(\calX)}).
\end{equation}
If one defines the operator $A_{\mathrm{JP}}\in\Herm(\calX\otimes\calX_1\otimes\calX_2)$ as
\begin{equation}
A_{\mathrm{JP}} =\I_{\L(\calX)}\odot\I_{\L(\calX)} =\frac{1}{2}\sum_{i,j=1}^{\dim(\calX)} E_{i,j}\otimes \left( \sum_{k=1}^{\dim(\calX)} E_{i,k}\otimes E_{k,j}+E_{k,j}\otimes E_{i,k} \right),
\end{equation}
one recovers the (standard) Jordan product from Definition \ref{def:jordan-prod} by choosing $A=A_{\mathrm{JP}}$. That is, the (standard) Jordan product of linear maps $\Phi_1 \in \T(\calX, \calY_1)$ and $\Phi_1 \in \T(\calX, \calY_1)$ may be expressed as
\begin{equation}
\Phi_1\odot\Phi_2 = \Phi_1\odot_{A_{\mathrm{JP}}}\Phi_2.
\end{equation}
This choice of $A$ is not unique. Indeed, for any choice of nonzero Hermitian operator $X \in\Herm(\calX)$ satisfying $\Tr(X) = 0$, one may define
\begin{equation}
A = A_{\mathrm{JP}} + \I_{\calX}\otimes X \otimes X.
\end{equation}
This operator clearly satisfies \eqref{eq:Tr1ATr2A} and thus defines a generalized Jordan product that is distinct from the standard one $\odot_{A_{\mathrm{JP}}}$.

The generalized Jordan product as defined in Definition \ref{def:genJordProd} possesses all of the key properties outline in Section \ref{sec:JordanProdDef} that are satisfied by the standard Jordan product. In particular, for any choice of trace-preserving linear maps $\Phi_1 \in \T(\calX, \calY_1)$ and $\Phi_2 \in \T(\calX, \calY_2)$, one has
\begin{equation}
\Tr_{\calY_2}(J(\Phi_1\odot_A\Phi_2)) = J(\Phi_1) \qquad\text{and}\qquad\Tr_{\calY_1}(J(\Phi_1\odot_A\Phi_2)) = J(\Phi_2)
\end{equation}
and, moreover, the map $\Phi_1\odot_A\Phi_2$ is also trace preserving when $\Phi_1$ and $\Phi_2$ are trace preserving. In other words, Proposition~\ref{prop:jordan-trace-preserving} holds for generalized Jordan products as well. Indeed, making use of the assumption that $\Phi_2$ is trace preserving, one has that
\begin{equation}
\Tr_{\calY_2}(J(\Phi_1\odot_A\Phi_2)) = (\I_{\L(\calX)}\otimes\Phi_1)\bigl(\Tr_{\calY_2}(A)\bigr) = (\I_{\L(\calX)}\otimes\Phi_1)(J(\I_{\L(\calX)})) = J(\Phi_1).
\end{equation}
An analogous argument shows that $\Tr_{\calY_1}(J(\Phi_1\odot_A\Phi_2)) = J(\Phi_2)$ which follows from the assumption that $\Phi_1$ is trace preserving. Moreover, if $\Phi_1$ and $\Phi_2$ are Hermitian-preserving, then $\Phi_1\odot_A\Phi_2$ is also Hermitian-preserving.

Analogous to the (standard) Jordan product, a generalized Jordan product provides a useful condition to check to see if two channels are compatible.

\begin{proposition}
\label{prop:jordan-gen-compat}
Let $\Phi_1 \in \C(\calX, \calY_1)$ and $\Phi_2 \in \C(\calX, \calY_2)$ be two quantum channels. If there exists $A \in \Herm(\calX \otimes \calX_1 \otimes \calX_2)$, where $\calX = \calX_1 = \calX_2$, satisfying \eqref{eq:Tr1ATr2A} such that $\Phi_1\odot_{A}\Phi_2$ is completely positive, then $\Phi_1\odot_{A}\Phi_2$ compatibilizes $\Phi_1$ and $\Phi_2$.
\end{proposition}

\begin{proof}
This follows directly from observations in the previous paragraph.
\end{proof}

This suggests the following definition.

\begin{definition}[Jordan compatible]
We say that $\Phi_1 \in \C(\calX, \calY_1)$ and $\Phi_2 \in \C(\calX, \calY_2)$ are Jordan compatible if there exists an operator $A \in \Herm(\calX \otimes \calX_1 \otimes \calX_2)$ satisfying \eqref{eq:Tr1ATr2A} (where $\calX = \calX_1 = \calX_2$) such that $\Phi_1\odot_{A}\Phi_2$ is completely positive.
\end{definition}

Therefore, if two channels are Jordan compatible, then they are compatible.

\begin{remark}
In this work, we consider both the generalized Jordan product and the standard version.  We sometimes refer to (standard) Jordan compatibility or (generalized) Jordan compatibility to emphasize which one we mean.
\end{remark}
   
%%%%%%%%%%%%%%%%%%%%%%%%%%%

\subsubsection*{Generalizing Jordan products of matrices.}
Following the procedure of generalizing the Jordan product for linear maps, one may use similar ideas to generalize the Jordan product of operators in the following manner. For every Hermitian operator $A\in\Herm(\calX\otimes\calX\otimes\calX)$ having the form $A=A_{\mathrm{JP}}+\I_{\calX}\otimes X\otimes X$ for some fixed choice of Hermitian operator $X\in\Herm(\calX)$ satisfying $\Tr(X)=0$, one may define the \emph{generalized Jordan product of operators} with respect to $A$ as
\begin{equation}
\label{eq:genJordProdOperators}
B\odot_A C = B\odot C + \ip{X\otimes X}{B\otimes C}\I_{\calX}.
\end{equation}
This type of generalized Jordan product of operators provides a condition for checking if two POVMs are compatible. Suppose $\{M_1,\dots,M_m\}\subset\Pos(\calX)$ and $\{N_1,\dots,N_n\}\subset\Pos(\calX)$ are POVMs. If it is the case that
\begin{equation}
M_i\odot_A N_j \geq0
\end{equation}
for each pair of indices $i$ and $j$ (where $\odot_A$ is the generalized Jordan product as defined in \eqref{eq:genJordProdOperators}), then the POVMs $M$ and $N$ are compatible as the operators defined as $P_{i,j}=M_i\odot N_j$ necessarily compose a compatibilizing POVM.

%%%%%%%%%%%%%%%%%%%%%%%%%%% 

\subsection{Jordan product compatibility of channels}
It is natural to ask if the converse to the main result of the previous subsection also holds. That is, if two channels are compatible, are they necessarily Jordan compatible? The results of Proposition~\ref{prop:jordan-Delta-NaS}, Proposition~\ref{thm:jordan-PVM-NaS}, Lemma \ref{lemma:jordan-unitary-NaS} and Lemma \ref{lemma:jordan-constant-NaS} show that this is true if either channel is of a certain type. We now show that this is also true for other classes of channels as well.

In the following, we make use of the \emph{inverse map} (if it exists) of a linear map of the form $\Phi \in \T(\calX, \calY)$. If $\Phi$ is completely positive and invertible as a linear map, its inverse $\Phi^{-1}$ may not necessarily be completely positive. However, the following lemma shows that $\Phi^{-1}$ is necessarily trace preserving if $\Phi$ is trace preserving.

Note the following: let $\Phi \in \T(\calX, \calY)$ and assume that an inverse map $\Phi^{-1} \in \T(\calY, \calX)$ exists. Then the vector spaces $\L(\calX)$ and $\L(\calY)$ are isomorphic and so also $\calX$ and $\calY$ are isomorphic.

\begin{lemma}
\label{lemma:jordan-gen-inverse-TP}
Let $\Phi \in \T(\calX, \calY)$ be an invertible linear map. The map $\Phi$ is trace preserving if and only if its inverse map $\Phi^{-1} \in \T(\calY, \calX)$ is trace preserving.
\end{lemma}
\begin{proof}
Suppose that $\Phi$ is trace preserving. For every $Y \in \L(\calY)$, one has that
\begin{equation}
\Tr( \Phi^{-1}(Y) ) = \Tr( \Phi( \Phi^{-1}(Y) ) ) = \Tr(Y) ,
\end{equation}
where the first equality follows from the assumption that $\Phi$ is trace preserving. The rest of the proof follows by symmetry between $\Phi$ and $\Phi^{-1}$.
\end{proof}

We now state the equivalence of channel compatibility and Jordan compatibility for certain pairs of channels.

\begin{theorem}
\label{thm:jordan-gen-inverseIFF}
Let $\Phi_1 \in \C(\calX, \calY_1)$ and $\Phi_2 \in \C(\calX, \calY_2)$ be two quantum channels, such that they have \emph{inverse} linear maps $\Phi_1^{-1} \in \T(\calY_1, \calX)$ and $\Phi_2^{-1} \in \T(\calY_2, \calX)$. The channels $\Phi_1$ and $\Phi_2$ are compatible if and only if they are Jordan compatible.
\end{theorem}
\begin{proof}
By Proposition~\ref{prop:jordan-gen-compat}, if the channels are Jordan compatible then they are compatible. To prove the converse, assume the channels are compatible and let $\Phi \in \C(\calX, \calY_1 \otimes \calY_2)$ be a compatibilizing channel. Define the operator $A\in\Herm(\calX\otimes\calX_1\otimes\calX_2)$, where $\calX=\calX_1=\calX_2$, as
\begin{equation}
A = J\bigl((\Phi_1^{-1}\otimes\Phi_2^{-1})\circ\Phi)\bigr).
\end{equation}
We now show that $A$ satisfies \eqref{eq:Tr1ATr2A}. Note that
\begin{align}
\Tr_{\calX_1}(A)
&= J\bigl(((\mathrm{Tr}\circ\Phi_1^{-1})\otimes\Phi_2^{-1}) \circ \Phi \bigr)\\
&= J\bigl((\mathrm{Tr}\otimes\Phi_2^{-1})\circ\Phi)\bigr)\\
&= J\bigl(\Phi_2^{-1}\circ(\mathrm{Tr}_{\calY_1}\circ\Phi)\bigr) \\
& = J\bigl(\Phi_2^{-1}\circ\Phi_2\bigr) \\
& = J(\I_{\L(\calX)}),
\end{align}
where equality in the second line follows from the fact that $\Phi_1^{-1}$ is trace preserving by Lemma \ref{lemma:jordan-gen-inverse-TP}. Similarly, one finds that
\begin{equation}
\Tr_{\calX_2}(A) = J\bigl(\Phi_1^{-1}\otimes\mathrm{Tr})\circ\Phi)\bigr)=J(\Phi_1^{-1}\circ\Phi_1) = J(\I_{\L(\calX)}),
\end{equation}
which completes the proof. Note that $J(\Phi_1 \odot_A \Phi_2) \geq 0$ simply because
\begin{equation}
J(\Phi_1 \odot_A \Phi_2) = (\I_{\L(\calX)} \otimes \Phi_1 \otimes \Phi_2)(A) = (\I_{\L(\calX)} \otimes \Phi_1 \otimes \Phi_2) ( J\bigl((\Phi_1^{-1}\otimes\Phi_2^{-1})\circ\Phi)\bigr) ) = J(\Phi).
\end{equation}
\end{proof}

Almost all linear maps in $\T(\calX)$ are invertible. Hence naive numerical approaches cannot be used to search for a pair of channels $\Phi_1, \Phi_2 \in \C(\calX)$ that are compatible but not Jordan compatible (if such a pair exists), as this would involve randomly sampling from a set having zero measure. For this reason, the question of whether compatibility of channels is equivalent to Jordan compatibility remains open. We conjecture that such a pair of channels does not exists. Nonetheless, it can be shown that the set of Jordan-compatible pairs of channels is dense in the set of all pairs of compatible channels, as will be discussed in section \ref{section:geometry}.

%%%%%%%%%%%%%%%%%%%%%%%%%%%%%%%%%%%%%%%% 

\section{Geometry of pairs of compatible quantum channels}
\label{section:geometry}

In this section, we discuss the geometry of the set of compatible pairs of channels. To this end, we introduce some notation. Let $\op{HP}(\calX,\calY)$ denote the space of Hermitian-preserving linear maps from $\L(\calX)$ to $\L(\calY)$. This is a real vector space with dimension
\begin{equation}
\op{dim}\bigl(\op{HP}(\calX,\calY)\bigr)=\op{dim}(\calX)^2\op{dim}(\calY)^2.
\end{equation}
Here we are concerned with the space $\op{HP}(\calX,\calY_1)\oplus\op{HP}(\calX,\calY_2)$ of pairs of such linear maps.

\begin{definition}
We define the following sets of pairs of Hermitian-preserving linear maps:
\begin{itemize}
\item $\op{CPairs}(\calX,\calY_1,\calY_2) = \{(\Phi_1,\Phi_2)\,:\, \Phi_1\in\op{C}(\calX,\calY_1),\, \Phi_2\in\C(\calX,\calY_2)\}$
\item $\op{Comp}(\calX,\calY_1,\calY_2) = \{(\Phi_1,\Phi_2)\in\op{CPairs}\,:\, \Phi_1\text{ and }\Phi_2\text{ are compatible}\}$
\item $\op{JComp}(\calX,\calY_1,\calY_2) = \{(\Phi_1,\Phi_2)\in\op{CPairs}\,:\, \Phi_1\text{ and }\Phi_2\text{ are Jordan compatible}\}$.
\end{itemize}
\end{definition}

It is evident that we have the containments
\begin{equation}
\op{JComp}(\calX,\calY_1,\calY_2)\subseteq\op{Comp}(\calX,\calY_1,\calY_2)\subsetneq \op{CPairs}(\calX,\calY_1,\calY_2).
\end{equation}
The remainder of this section is dedicated to stating and proving a few facts regarding the geometry of these sets, which we summarize here to outline our approach:
\begin{enumerate}
\item The set $\op{Comp}(\calX,\calY_1,\calY_2)$ is compact and convex.
\item The set $\op{Comp}(\calX,\calY_1,\calY_2)$ has positive measure as a subset of $\op{CPairs}(\calX,\calY_1,\calY_2)$. (That is, a randomly selected pair of channels has a nonzero probability of being compatible.)
\item In the case when $\calX=\calY_1=\calY_2$, almost all pairs of compatible channels are Jordan compatible. (That is, the set of non-Jordan-compatible pairs has zero measure as a subset of $\op{Comp}(\calX,\calX,\calX)$).

\item In particular, one has that $\overline{\op{JComp}(\calX,\calX,\calX)} = \op{Comp}(\calX,\calX,\calX)$ (where $\overline{\mathcal{A}}$ indicates the topological closure of a set $\mathcal{A}$).
\end{enumerate}

We stress that the above implies that with probability $1$, a randomly selected pair of compatible channels (with respect to the measure as discussed below) is also Jordan compatible. This rules out, for instance, a brute-force random search to find a pair of channels that are compatible but not Jordan compatible. 
   
Since we are referring to probabilities and measure-zero sets of pairs of channels, it is necessary to clarify the measure on the set $\op{CPairs}(\calX,\calY_1,\calY_2)$. Note that $\op{CPairs}(\calX,\calY_1,\calY_2)$ is a compact and convex subset of the affine subspace of pairs of trace-preserving maps in the Euclidean space $\op{HP}(\calX,\calY_1)\oplus\op{HP}(\calX,\calY_2)$. Thus the set $\op{CPairs}$ is a submanifold and has a measure that is induced by the natural Lebesgue measure of the underlying Euclidean space. (See, e.g., \cite[Section 5.5]{Lerner2014}.)

\subsubsection*{Convexity of the set of pairs of compatible channels}

Convexity of the set $\op{CPairs}(\calX,\calY_1,\calY_2)$ follows trivially from convexity of $\C(\calX,\calY_1)$ and $\C(\calX,\calY_2)$. Importantly, the set of compatible pairs of channels is also convex---a fact that we prove in the following lemma.
\begin{lemma}
\label{lemConvexity}
The set $\op{Comp}(\calX,\calY_1,\calY_2)$ is convex.
\end{lemma}
\begin{proof}
Let $(\Phi_1,\Phi_2)$ and $(\Psi_1,\Psi_2)$ be pairs of compatible channels, let $\Phi,\Psi\in\C(\calX,\calY_1\otimes\calY_2)$ be respective channels that compatibilize these pairs, and let $\lambda\in[0,1]$. One has that
\begin{equation}
\Tr_{\calY_2}\circ(\lambda\Phi+(1-\lambda)\Psi)= \lambda\Tr_{\calY_2} \circ \Phi + (1-\lambda)\Tr_{\calY_2} \circ \Psi = \lambda\Phi_1 + (1-\lambda)\Psi_1
\end{equation}
and, analogously, that $\Tr_{\calY_1}\circ(\lambda\Phi+(1-\lambda)\Psi) = \lambda\Phi_2 + (1-\lambda)\Psi_2$. It follows that $\lambda\Phi+(1-\lambda)\Psi$ is a channel that compatibilizes the pair
\begin{equation}
\lambda(\Phi_1,\Phi_2)+(1-\lambda)(\Psi_1,\Psi_2) = \bigl(\lambda\Phi_1+(1-\lambda)\Psi_1,\,\lambda\Phi_2+(1-\lambda)\Psi_2\bigr)
\end{equation}
and thus this pair is compatible.
\end{proof}

We now prove the following useful result, which states that mixing any pair of channels with a pair of constant channels yields a compatible pair.

\begin{proposition}
\label{prop:Phi1Phi2-mixed-with-Psi1Psi2}
Let $\rho_1\in\D(\calY_1)$ and $\rho_2\in\D(\calY_2)$ be states and let $\Phi_{\rho_1}\in\C(\calX,\calY_1)$ and $\Phi_{\rho_2}\in\C(\calX,\calY_2)$ be the constant channels defined as
\begin{equation}
\Phi_{\rho_1}(X) = \Tr(X)\rho_1\qquad\text{and}\qquad \Phi_{\rho_2}(X) = \Tr(X)\rho_2
\end{equation}
for every $X\in\L(\calX)$. For every other pair of channels $\Psi_1\in\C(\calX,\calY_1)$ and $\Psi_2\in\C(\calX,\calY_2)$, it holds that
\begin{equation}
\left(\frac{1}{2}\Psi_1+\frac{1}{2}\Phi_{\rho_1},\, \frac{1}{2}\Psi_2+\frac{1}{2}\Phi_{\rho_2}\right)\in\op{Comp}(\calX,\calY_1,\calY_2).
\end{equation}
\end{proposition}

\begin{proof}
\label{lem:halfway-pairs-are-compatible}
It is evident that the constant channels $\Phi_{\rho_1}$ and $\Phi_{\rho_2}$ are each compatible with every channel in $\C(\calX,\calY_2)$ and $\C(\calX,\calY_1)$ (see Lemma \ref{lemma:jordan-constant-NaS}). In particular, one has that
\begin{equation}
(\Psi_1,\Phi_{\rho_2})\in\op{Comp}(\calX,\calY_1,\calY_2)
\qquad\text{and}\qquad
(\Phi_{\rho_1},\Psi_2)\in\op{Comp}(\calX,\calY_1,\calY_2).
\end{equation}
It follows from convexity (Lemma~\ref{lemConvexity}) that the pair
\begin{equation}
\frac{1}{2}(\Psi_1,\Phi_{\rho_2}) + \frac{1}{2}(\Phi_{\rho_1},\Psi_2) =
\left(\frac{1}{2}\Psi_1+\frac{1}{2}\Phi_{\rho_1},\, \frac{1}{2}\Psi_2+\frac{1}{2}\Phi_{\rho_2}\right)
\end{equation}
is compatible.
\end{proof}

Importantly, we also point out that the set of compatible pairs is also closed, as proved below.

\begin{proposition}
\label{prop:comp-is-closed}
The set $\op{Comp}(\calX,\calY_1,\calY_2)$ is compact.
\end{proposition}

\begin{proof}
This may be proved by observing that $\op{Comp}(\calX,\calY_1,\calY_2)$ is the image of the set of channels $\C(\calX,\calY_1\otimes\calY_2)$ under the linear map $\op{HP}(\calX,\calY_1\otimes\calY_2)\rightarrow\op{HP}(\calX,\calY_1)\oplus \op{HP}(\calX,\calY_2)$ that is defined by $\Phi\mapsto({\Tr_{\calY_2}}\circ\Phi,{\Tr_{\calY_1}}\circ\Phi)$. The desired result follows from the fact that $\C(\calX,\calY_1\otimes\calY_2)$ is compact.
\end{proof}

\subsubsection*{Norm and measure for pairs of channels}

We may define a norm on the real vector space $\op{HP}(\calX,\calY_1)\oplus\op{HP}(\calX,\calY_2)$ as follows. For a pair $(\Phi_1,\Phi_2)$ of Hermitian-preserving maps, define
\begin{equation}
\lVert (\Phi_1,\Phi_2)\rVert_J = 
\lVert J(\Phi_1)\rVert + \lVert J(\Phi_2)\rVert
\end{equation}
where $\lVert J(\Phi_1)\rVert$ and $\lVert J(\Phi_2)\rVert$ denote the operator norms of the Choi representations of $\Phi_1$ and $\Phi_2$. The following lemma shows the existence of a ball of positive radius that is fully contained in the set of pairs of channels.

\begin{lemma}
\label{lem:ball-pairs-of-channels}
Let $\Omega_{\calY_1}\in\C(\calX,\calY_1)$ and $\Omega_{\calY_2}\in\C(\calX,\calY_2)$ be the constant channels defined as
\begin{equation}
\label{eq:Psi1Psi2depolarizing}
\Omega_{\calY_1}(X) = \frac{\Tr(X)}{\op{dim}(\calY_1)}\I_{\calY_1}\qquad\text{and}\qquad \Omega_{\calY_2}(X) = \frac{\Tr(X)}{\op{dim}(\calY_2)}\I_{\calY_2}
\end{equation}
for all $X\in\L(\calX)$ and let $(\Phi_1,\Phi_2)\in\op{HP}(\calX,\calY_1)\oplus\op{HP}(\calX,\calY_1)$ be a pair of Hermitian-preserving maps for which
\begin{equation}
\label{eq:Jnorm_of_difference_of_pairs_mindy1dy2}
\big\lVert(\Phi_1,\Phi_2) - (\Omega_{\calY_1},\Omega_{\calY_2})\big\rVert_J \leq \min\left\{\frac{1}{\op{dim}(\calY_1)},\frac{1}{\op{dim}(\calY_2)}\right\}.
\end{equation}
If $\Phi_1$ and $\Phi_2$ are also trace preserving, then $(\Phi_1,\Phi_2)\in\op{CPairs}(\calX,\calY_1,\calY_2)$.
\end{lemma}
\begin{proof}
Suppose that the Hermitian-preserving linear maps $\Phi_1$ and $\Phi_2$ are trace preserving. It suffices to show that $\Phi_1$ and $\Phi_2$ are completely positive. Note that the inequality in \eqref{eq:Jnorm_of_difference_of_pairs_mindy1dy2} implies that
\begin{equation}
\label{eq:JPhi-minus-Psi}
\lVert J(\Phi_1)-J(\Omega_{\calY_1})\rVert \leq \frac{1}{\op{dim}(\calY_1)}
\qquad\text{and}\qquad
\lVert J(\Phi_2)-J(\Omega_{\calY_2})\rVert \leq \frac{1}{\op{dim}(\calY_2)}.
\end{equation}
Observe that
\begin{equation}
\label{eq:JPsi1-JPsi2}
J(\Omega_{\calY_1}) = \frac{1}{\op{dim}(\calY_1)}\I_\calX\otimes\I_{\calY_1} \qquad\text{and}\qquad J(\Omega_{\calY_1}) = \frac{1}{\op{dim}(\calY_2)}\I_\calX\otimes\I_{\calY_2}.
\end{equation}
The inequalities in \eqref{eq:JPhi-minus-Psi}, together with the equalities in \eqref{eq:JPsi1-JPsi2}, imply that $J(\Phi_1)\geq0$ and $J(\Phi_2)\geq0$, and thus $\Phi_1$ and $\Phi_2$ are completely positive, as required.
\end{proof}

For the channels $\Omega_{\calY_1}$ and $\Omega_{\calY_2}$ defined in \eqref{eq:Psi1Psi2depolarizing}, Lemma \ref{lem:ball-pairs-of-channels} implies that, for any pair $(\Phi_1,\Phi_2)$ of trace-preserving linear maps that is within a distance of $\min\{1/\op{dim}(\calY_1),1/\op{dim}(\calY_2)\}$ from the pair of channels $(\Omega_{\calY_1},\Omega_{\calY_2})$, the maps $\Phi_1$ and $\Phi_2$ are also channels. This fact, together with the result from Proposition \ref{prop:Phi1Phi2-mixed-with-Psi1Psi2}, implies the following proposition.

\begin{proposition}
Let $\Phi_1\in\C(\calX,\calY_1)$ and $\Phi_2\in\C(\calX,\calY_2)$ be channels, let $\Omega_{\calY_1}$ and $\Omega_{\calY_2}$ be the channels as defined in Equation~\eqref{eq:Psi1Psi2depolarizing}, and suppose that
\begin{equation}
\lVert (\Phi_1,\Phi_2) - (\Omega_{\calY_1},\Omega_{\calY_2})\rVert_J \leq \frac{1}{2}\min\left\{\frac{1}{\op{dim}(\calY_1)},\frac{1}{\op{dim}(\calY_2)}\right\}.
\end{equation}
Then the channels $\Phi_1$ and $\Phi_2$ are compatible.
\end{proposition}

\begin{proof}
Consider the linear maps defined as $\Psi_1 = 2\Phi_1 - \Omega_{\calY_2}$ and $\Psi_2 = 2\Phi_2 - \Omega_{\calY_2}$. It is evident that these maps are trace-preserving. We have that
\begin{equation}
\lVert (\Psi_1,\Psi_2) - (\Omega_{\calY_1},\Omega_{\calY_2})\rVert_J = 2\lVert (\Phi_1,\Phi_2)-(\Omega_{\calY_1},\Omega_{\calY_2})\rVert_J \leq \min\left\{\frac{1}{\op{dim}(\calY_1)},\frac{1}{\op{dim}(\calY_2)}\right\},
\end{equation}
and thus $\Psi_1$ and $\Psi_2$ are channels by Lemma \ref{lem:ball-pairs-of-channels}. It follows from Proposition \ref{prop:Phi1Phi2-mixed-with-Psi1Psi2} that the pair
\begin{equation}
(\Phi_1,\Phi_2) = \frac{1}{2}(\Psi_1,\Psi_2) + \frac{1}{2}(\Omega_{\calY_1},\Omega_{\calY_2})
\end{equation}
is compatible, as desired.
\end{proof}

The existence of a ball of positive radius within the set of all pairs of compatible channels implies the following corollary.

\begin{corollary}
The {sets $\op{Comp}(\calX,\calY_1,\calY_2)$ and $\op{CPairs}(\calX,\calY_1,\calY_2)$ have the same dimension as convex sets}. In particular, $\op{Comp}(\calX,\calY_1,\calY_2)$ has positive measure in $\op{CPairs}(\calX,\calY_1,\calY_2)$.
\end{corollary}

In particular, this means that a randomly selected pair of channels has a nonzero probability of being compatible.

\subsubsection*{Invertible maps and Jordan compatibility}

Here we prove that almost all pairs of compatible pairs of channels are Jordan compatible. To do so, we first observe that almost all channels are invertible as linear maps (and thus almost all pairs of channels are invertible pairs).

\begin{lemma}
Almost all channels in $\C(\calX)$ are invertible as linear maps (in the sense that non-invertible channels form a set of measure zero).
\end{lemma}

\begin{proof}
We may view all Hermitian-preserving maps on $\L(\calX)$ as linear maps on the real vector space of Hermitian operators $\op{Herm}(\calX)$. By setting $n=\op{dim}(\calX)$, we may identify $\op{Herm}(\calX)\simeq \mathbb{R}^{n^2}$, and we may identify the space $\op{HP}(\calX)$ of Hermitian-preserving maps with the space of $n^2\times n^2$ matrices over $\mathbb{R}$. Moreover, under these identifications, the affine space of trace-preserving maps in $\op{HP}(\calX)$ corresponds to some affine subspace of $n^2\times n^2$ matrices. The set $\C(\calX)$ of all channels may be identified with a convex subset in this affine subspace. Finally, note that the determinant is a polynomial on the $n^4$-dimensional vector space of $n^2\times n^2$ matrices over $\mathbb{R}$, and thus the determinant is either constant on this affine subspace or the set of zeroes has measure zero (with respect to the measure on this affine subspace that is induced by the Lebesgue measure on $\mathbb{R}^{n^4}$). A matrix is invertible if and only if its determinant is nonzero. Since the identity channel is certainly invertible, it follows from the above observations that the set of non-invertible channels (as a subset of all channels) has zero measure.
\end{proof}

This immediately implies that almost all pairs of channels are invertible.

\begin{corollary}
For almost all pairs of channels $(\Phi_1,\Phi_2)\in\op{CPairs}(\calX,\calX,\calX)$, the channels $\Phi_1$ and $\Phi_2$ are invertible as linear maps. (That is, the set of pairs of channels such that at least one channel is not invertible has zero measure.)
\end{corollary}

We may now prove the main result, which is that almost all pairs of compatible channels are Jordan compatible. In particular, this means that the set of compatible channels is equal to the closure of the set of Jordan compatible channels.

\begin{theorem}
\label{thm:geometry-fullMeasure}
The following statements hold.
\begin{enumerate}[label = (\arabic*)]
\item\label{item:geometry-fullMeasure-1} The set $\op{JComp}(\calX,\calX,\calX)$ has full measure as a subset of $\op{Comp}(\calX,\calX,\calX)$.
\item\label{item:geometry-fullMeasure-2} $\overline{\op{JComp}(\calX,\calX,\calX)} = \op{Comp}(\calX,\calX,\calX)$.
\end{enumerate}
\end{theorem}

\begin{proof}
Consider the set of pairs of invertible channels, which we denote as
\begin{equation}
\op{InvPairs}(\calX,\calX,\calX) = \{(\Phi_1,\Phi_2)\in\op{CPairs}(\calX,\calX,\calX)\,:\, \Phi_1,\Phi_2\text{ are invertible as linear maps}\}.
\end{equation}
Recall from Theorem \ref{thm:jordan-gen-inverseIFF} that, a pair of invertible channels $(\Phi_1,\Phi_2)\in\op{InvPairs}(\calX,\calX,\calX)$ is compatible if and only if it is Jordan compatible. It follows that
\begin{equation}
\op{Comp} \cap \op{InvPairs}\subseteq \op{JComp} \subseteq \op{Comp}
\end{equation}
(where, for simplicity, we have left off the ``$(\calX,\calX,\calX)$'' part of each set in the above containments). Statement~\ref{item:geometry-fullMeasure-1}~now follows from the facts that $\op{Comp}$ has positive measure in $\op{CPairs}$ and the set $\op{InvPairs}$ has full measure in the set $\op{CPairs}$. Finally, statement~\ref{item:geometry-fullMeasure-2}~is a trivial corollary of statement~\ref{item:geometry-fullMeasure-1}, as the set $\op{Comp}$ is closed (see Proposition \ref{prop:comp-is-closed}).
\end{proof}

%%%%%%%%%%%%%%%%%%%%%%%%%%% 
  
\section{Semidefinite programming (SDP) formulations}
\label{SectSDP}
   
In this section, we formulate some of the compatibility questions posed in this work as semi\-definite programs and examine them in a new light via duality theory. Also, we use semidefinite programming methods to provide a novel proof the well-known result that there is no perfect broadcasting in quantum theory.

%%%%%%%%%%%%%%%%%%%%%%%%%%% 

\subsection{SDP formulations of compatibility}

Recall from the introduction that determining whether a given pair of channels $\Phi_1 \in \C(\calX,\calY_1)$ and $\Phi_2 \in \C(\calX,\calY_2)$ is compatible is equivalent to solving the following feasibility problem:
\begin{equation}
\label{CPcomp}
\begin{split}
\text{find:} \quad & \Phi \text{ completely positive} \\
\text{satisfying:}\quad & 
\Phi_1 = \Tr_{\calY_2} \circ \; \Phi \\
& \Phi_2 = \Tr_{\calY_1} \circ \; \Phi.
\end{split}
\end{equation}
Using the Choi representations, one finds that this is equivalent to the following semidefinite programming feasibility problem:
\begin{equation}
\label{eq:sdpchannelcompatibility1}
\begin{split}
\text{find:} \quad & X \in \Pos(\calX \otimes \calY_1 \otimes \calY_2) \\
\text{satisfying:} \quad & \Tr_{\calY_2}(X) = J(\Phi_1) \\
& \Tr_{\calY_1}(X) = J(\Phi_2),
\end{split}
\end{equation}
where a solution $X$ to the problem \eqref{eq:sdpchannelcompatibility1} is the Choi representation of a compatibilizing channel in~\eqref{CPcomp} (if one exists). In other words, we are translating the channel problem into a matrix problem via the Choi isomorphism.

The channel compatibility problem can also be phrased in terms of the following pair of semidefinite programs. Let $\alpha_{\comp}$ be the optimal value of the following semidefinite program
\begin{equation}
\label{eq:sdpchannelcompatibility}
\begin{split}
\text{maximize:} \quad & t \\
\text{satisfying:} \quad & \Tr_{\calY_2}(X) = J(\Phi_1) \\
& \Tr_{\calY_1}(X) = J(\Phi_2) \\
& X \geq t \cdot \I_{\calX \otimes \calY_1 \otimes \calY_2}
\end{split}
\end{equation}
and let $\beta_{\comp}$ be the optimal value of its dual problem, which can be stated as
\begin{equation}
\begin{split}
\label{eq:sdpchannelcompatibilitydual}
\text{minimize:} \quad
& \ip{Z_1}{J(\Phi_1)} + \ip{Z_2}{J(\Phi_2)} \\
\text{satisfying:} \quad
& \Tr_{\calY_2}^*(Z_1) + \Tr_{\calY_1}^*(Z_2) \in \D(\calX \otimes \calY_1 \otimes \calY_2) \\
& Z_1 \in \Herm(\calX\otimes\calY_1) \\
& Z_2 \in \Herm(\calX\otimes\calY_2).
\end{split}
\end{equation}
Recall here that $\Tr_{\calY_1}^*$ and $\Tr_{\calY_2}^*$ are the adjoints of the partial trace maps that are defined in \eqref{eq:Tr*definition}.

Strong duality holds for this pair of semidefinite programs, as we now argue. Indeed, the operator
\begin{equation}
\label{eq}
\bar{X} =
\frac{1}{\dim(\calY_2)}\Tr_{\calY_2}^*\bigl(J(\Phi_1)\bigr) + \frac{1}{\dim(\calY_1)}\Tr_{\calY_1}^*\bigl(J(\Phi_2)\bigr) - \frac{1}{\dim(\calY_1 \otimes \calY_2)} \I_{\calX \otimes \calY_1 \otimes \calY_2}
\end{equation}
is a feasible solution for the primal SDP, and thus
\begin{equation}
\alpha_{\comp} \geq \lambda_{\min}(\bar{X}).
\end{equation}
Similarly, the operators
\begin{equation}
\bar{Z}_1 := \frac{1}{2 \dim(\calX \otimes \calY_1 \otimes \calY_2)} \I_{\calX \otimes \calY_1}
\quad
\text{ and }
\quad
\bar{Z}_2 := \frac{1}{2 \dim(\calX \otimes \calY_1 \otimes \calY_2)} \I_{\calX \otimes \calY_2}
\end{equation}
form a strictly feasible dual solution. Thus,
\begin{equation}
\beta_{\comp} \leq \frac{1}{\dim(\calY_1 \otimes \calY_2)}.
\end{equation}
{By Slater's theorem (see, e.g., \cite[Theorem 1.18]{Watrous-QI})} it holds that $\alpha_{\comp} = \beta_{\comp}$ and that the optimal value $\alpha_{\comp}$ is attained. Summarizing these bounds, we have
\begin{equation}
\label{SDPbounds}
\lambda_{\min}(\bar{X}) \leq \alpha_{\comp} = \beta_{\comp} \leq \frac{1}{\dim(\calY_1 \otimes \calY_2)}.
\end{equation}
Note the fact that $\alpha_{\comp}$ is attained tells us something interesting. It tells us that $\Phi_1$ and $\Phi_2$ are compatible if and only if $\alpha_{\comp} \geq 0$ (and thus $\beta_{\comp} \geq 0$). This brings us to the following theorem.

\begin{theorem}[Theorem of the alternative (version 1)]
\label{TotAV1}
\emph{Exactly one} of the following statements is true:
\begin{enumerate}[label = (\arabic*)]
\item\label{item:TotAV1-1} $\Phi_1 \in \C(\calX,\calY_1)$ and $\Phi_2 \in \C(\calX,\calY_2)$ are compatible.
\item\label{item:TotAV1-2} There exists $Z_1 \in \Herm(\calX \otimes \calY_1)$ and $Z_2 \in \Herm(\calX \otimes \calY_2)$ such that
 \begin{equation}
 \label{eq:dual-like}  
 \Tr_{\calY_2}^*(Z_1) + \Tr_{\calY_1}^*(Z_2) \geq 0
 \quad \text{ and } \quad \ip{Z_1}{J(\Phi_1)} + \ip{Z_2}{J(\Phi_2)} < 0.
 \end{equation}
\end{enumerate}
\end{theorem}

\begin{proof}
As stated previously, statement~\ref{item:TotAV1-1} is equivalent to the condition that $\alpha_{\comp} \geq 0$. Thus, if statement~\ref{item:TotAV1-1} is not true, then $\alpha_{\comp} < 0$ and thus there exists a dual feasible solution $(\bar{Z}_1, \bar{Z}_2)$ with negative objective function value. The pair $(\bar{Z}_1, \bar{Z}_2)$ witnesses that statement~\ref{item:TotAV1-2} is true. In other words, both statements cannot be false.

Now, suppose that both statements are true for the purpose of a contradiction. This implies the existence of a primal feasible solution $\bar{X} \geq 0$ and $Z_1 \in \Herm(\calX \otimes \calY_1)$ and $Z_2 \in \Herm(\calX \otimes \calY_2)$ such that the conditions in \eqref{eq:dual-like} hold. Now we have 
\begin{align}
0 & > \ip{Z_1}{J(\Phi_1)} + \ip{Z_2}{J(\Phi_2)} \\ & = \ip{Z_1}{\Tr_{\calY_2}(\bar{X})} + \ip{Z_2}{\Tr_{\calY_1}(\bar{X})} \\ & = \ip{\Tr_{\calY_2}^*(Z_1) + \Tr_{\calY_1}^*(Z_2)}{\bar{X}} \\ & \geq 0,
\end{align}
as each operator is positive semidefinite, which yields a contradiction. Thus, both statements cannot be true. The result follows.
\end{proof}

To state a neat corollary, we define an inner product on linear maps.

\begin{definition}
We may define an inner product on $\T(\calX,\calY)$ as
\begin{equation}
\ip{\Psi}{\Phi} := \ip{J(\Psi)}{J(\Phi)}
\end{equation}
for every choice of linear maps $\Phi, \Psi \in \T(\calX,\calY)$. That is, the inner product between two linear maps is defined as the inner product of their Choi representations. Note that this is a proper inner product as the Choi representation is an isomorphism.
\end{definition}

We now have the following corollary of Theorem~\ref{TotAV1}.
\begin{theorem}[Theorem of the alternative (version 2)]

\emph{Exactly one} of the following statements is true:
\begin{enumerate}[label = (\arabic*)]
\item $\Phi_1 \in \C(\calX,\calY_1)$ and $\Phi_2 \in \C(\calX,\calY_2)$ are compatible.
\item There exists Hermitian-preserving maps $\Psi_1 \in \T(\calX,\calY_1)$ and $\Psi_2 \in \T(\calX,\calY_2)$ such that
 \begin{equation}
 \Tr_{\calY_2}^* \circ \Psi_1 + \Tr_{\calY_1}^* \circ \Psi_2 \; \text{ is completely positive} 
 \quad \text{ and } \quad
 \ip{\Psi_1}{\Phi_1} + \ip{\Psi_2}{\Phi_2} < 0.
 \end{equation}
\end{enumerate}
\end{theorem}

We now use the machinery we have developed to provide a novel proof of the no-broadcasting theorem in the following example.

\begin{example}[No-broadcasting theorem] \label{exm:no-broadcasting}
Recall that the no-broadcasting theorem states that the identity channel is not self-compatible. To study the self-compatibility of the identity channel, we define the \emph{partially depolarizing channel} $\Omega_p \in \C(\calX)$ with parameter $p\in[0,1]$, as
\begin{equation}
\label{eq:partialdepolarising}
\Omega_p = p \Omega + (1-p) \I_{\L(\calX)}.
\end{equation}
Consider the operators $Z_1$ and $Z_2$ defined as
\begin{equation}
Z_1 = Z_2 := \I\otimes \I - \dfrac{2}{{\op{dim}(\calX)}+1}\sum_{i,j=1}^{\op{dim}(\calX)} E_{i,j}\otimes E_{i,j}.
\end{equation}
It can be verified that this choice of operators satisfies
\begin{equation}
\Tr_{\calY_2}^*(Z_1) + \Tr_{\calY_1}^*(Z_2)\geq0
\qquad\text{and}\qquad \ip{Z_1+Z_2}{J(\Omega_p)} <0
\end{equation}
for $p$ in the range
\begin{equation}
0 \leq p < \dfrac{{\op{dim}(\calX)}}{2({\op{dim}(\calX)}+1)}.
\end{equation}
Thus, by Theorem \ref{TotAV1}, the channel $\Omega_p$ is not self-compatible for these values of $p$ (see also~\cite{Werner-cloning, KeylWerner-cloning} for another proof of this fact). Taking $p=0$ implies the result of the no-broadcasting theorem, as $\Omega_0$ is the identity channel.
\end{example}

\begin{remark}[Rewriting the dual using Jordan products]
\label{JordanRemark} We have shown a few close relationships between the notions of compatibility and Jordan products. We now show that this relationship is rather natural, as the Jordan product appears (in a somewhat hidden form) in the dual SDP above.
To see this, observe that
\begin{equation}
\ip{Z_1}{J(\Phi_1)} + \ip{Z_2}{J(\Phi_2)}
= \ip{J(\Phi_1 \odot_A \Phi_2)}{\Tr_{\calY_2}^*(Z_1) + \Tr_{\calY_1}^*(Z_2) }
\end{equation}
where $J(\Phi_1 \odot_A \Phi_2)$ is a generalized Jordan product for every choice of operator $A$ satisfying the conditions in \eqref{eq:Tr1ATr2A}. From this we may obtain the following alternative form of the dual problem in \eqref{eq:sdpchannelcompatibilitydual}: 
\begin{equation}
\label{eq:dualsdpchannelJordan}
\begin{split}
\text{minimize:} \quad & \ip{J(\Phi_1 \odot_A \Phi_2)}{\rho} \\
\text{subject to:} \quad 
& \rho = \Tr_{\calY_2}^*(Z_1) + \Tr_{\calY_1}^*(Z_2) \in \D(\calX \otimes \calY_1 \otimes \calY_2)\\
& Z_1 \in \Herm(\calX \otimes \calY_1) \\
& Z_2 \in \Herm(\calX \otimes \calY_2).
\end{split}
\end{equation} 
Note that this yields a previous result of ours. Namely, if there exists an operator $A$ satisfying the conditions in \eqref{eq:Tr1ATr2A} such that $J(\Phi_1 \odot_A \Phi_2) \geq 0$, then from the dual (as written above) we have that $\beta_{\comp} \geq 0$ and thus $\Phi_1$ and $\Phi_2$ are compatible.
\end{remark}

%%%%%%%%%%%%%%%%%%%%%%%%%%%%%%

\subsection{SDP for finding a generalized Jordan product compatibilizer}

Recall that for channels $\Phi_1 \in \C(\calX, \calY_1)$ and $\Phi_2 \in \C(\calX, \calY_2)$, the task of determining Jordan compatibility is equivalent to the following problem
\begin{equation}
\begin{split}
\text{find:} \quad
& A \in \Herm(\calX \otimes \calX_1 \otimes \calX_2) \\
\text{satisfying:} \quad
& \Tr_{\calX_1}(A) = J(\I_{\L(\calX)}) \\
& \Tr_{\calX_2}(A) = J(\I_{\L(\calX)}) \\
& \bigl( \I_{\L(\calX)} \otimes \Phi_1 \otimes \Phi_2 \bigr)(A) \in \Pos(\calX\otimes\calY_1\otimes\calY_2)
\end{split}
\end{equation}
where $J(\Phi_1 \odot_A \Phi_2) = \bigl(\I_{\L(\calX)}\otimes\Phi_1\otimes\Phi_2\bigr)(A)$ is the generalized Jordan product of $\Phi_1$ and $\Phi_2$ with respect to $A$.

The problem of determining Jordan compatibility of a pair of channels can be phrased in terms of the following pair of semidefinite programs. Let $\alpha_{\jord}$ be the optimal value of the following semidefinite program
\begin{equation}
\begin{split}
\label{eq:jordan-gen-SDP}
\text{maximize:} \quad & t \\
\text{satisfying:} \quad
& \Tr_{\calX_1}(A) = J(\I_{\L(\calX)}) \\
& \Tr_{\calX_2}(A) = J(\I_{\L(\calX)}) \\
& \bigl( \I_{\L(\calX)} \otimes \Phi_1 \otimes \Phi_2 \bigr)(A) \geq t \cdot \I_{\calX \otimes \calY_1 \otimes \calY_2} \\
& A \in \Herm(\calX \otimes \calX_1 \otimes \calX_2),
\end{split}
\end{equation}
and let $\beta_{\jord}$ be the optimal value of the corresponding dual problem, which can be stated as follows
\begin{equation}
\label{eq:sdpchannelcompatibilityJdual}
\begin{split}
\text{minimize:} \quad & \ip{W_1+W_2}{J(\I_{\L(\calX)})} \\
\text{satisfying:} \quad
& (\I_{\L(\calX)} \otimes \Phi_1^* \otimes \Phi_2^*)(\rho) =
\Tr_{\calX_2}^*(W_1) + \Tr_{\calX_1}^*(W_2) \\ 
& W_1 \in \Herm(\calX \otimes \calX_1) \\
& W_2 \in \Herm(\calX \otimes \calX_2) \\
& \rho \in \D(\calX \otimes \calY_1 \otimes \calY_2).
\end{split}
\end{equation}
We now show that strong duality holds for this pair of semidefinite programs. To see this, note that the operator $A = A_{\JP}$ together with the value $t = \lambda_{\min}(\bigl( \I_{\L(\calX)} \otimes \Phi_1 \otimes \Phi_2 \bigr)(A_{\JP}))$ forms a feasible solution to the primal problem in \eqref{eq:jordan-gen-SDP}, and thus
\begin{equation}
\alpha_{\jord} \geq \lambda_{\min} \left( \bigl( \I_{\L(\calX)} \otimes \Phi_1 \otimes \Phi_2 \bigr)(A_{\JP}) \right).
\end{equation}
Similarly, the operators
\begin{align}
W_1 & := \frac{1}{2 \dim(\calX)^3} \; \I_{\calX \otimes \calX_1} \\ W_2 & := \frac{1}{2 \dim(\calX)^3} \; \I_{\calX \otimes \calX_2} \\ \rho & := \frac{1}{\dim(\calX)^3} \; \I_{\calX \otimes \calY_1 \otimes \calY_2}
\end{align}
form a strictly feasible solution to the dual problem, and thus
\begin{equation}
\beta_{\jord} \leq \frac{1}{\dim(\calX)^2}.
\end{equation}
Strong duality now follows from the fact that the primal is feasible and the dual is strictly feasible. Moreover, it holds that $\alpha_{\jord} = \beta_{\jord}$ and that the optimal value $\alpha_{\jord}$ is attained. Summarizing these bounds, we have
\begin{equation}
\lambda_{\min} \left( \bigl( \I_{\L(\calX)} \otimes \Phi_1 \otimes \Phi_2 \bigr)(A_{\JP}) \right) \leq \alpha_{\jord} = \beta_{\jord} \leq \frac{1}{\dim(\calX)^2}.
\end{equation}
Again, note that the fact that $\alpha_{\jord}$ is attained tells us something interesting. It tells us that $\Phi_1$ and $\Phi_2$ are Jordan compatible if and only if $\alpha_{\jord} \geq 0$ (and thus $\beta_{\jord} \geq 0$). This brings us to the following theorem.

\begin{theorem}[Theorem of the alternative (Jordan version)]
\label{TotAVJ1}
\emph{Exactly one} of the following statements is true:
\begin{enumerate}[label = (\arabic*)]
\item $\Phi_1 \in \C(\calX,\calY_1)$ and $\Phi_2 \in \C(\calX,\calY_2)$ are Jordan compatible.
\item There exists $W_1 \in \Herm(\calX \otimes \calX_1)$, $W_2 \in \Herm(\calX \otimes \calX_2)$, and $\rho \in \Pos(\calX \otimes \calY_1 \otimes \calY_2)$ such that
 \begin{equation}
 \label{eq:phi1phi2starW1W2}
 (\I_{\L(\calX)} \otimes \Phi_1^* \otimes \Phi_2^*)(\rho) = \Tr_{\calX_2}^*(W_1)+ \Tr_{\calX_1}^*(W_2)\\ 
 \quad \text{ and } \quad
 \bigip{W_1+W_2}{J(\I_{\L(\calX)})} < 0.
 \end{equation}
\end{enumerate}
\end{theorem}

\begin{proof}
The proof is completely analogous to that of Theorem~\ref{TotAV1}.
\end{proof}

\begin{remark}[Rewriting this dual using Jordan products]
Again, we can use Jordan products to rewrite the dual SDP~\eqref{eq:sdpchannelcompatibilityJdual}. In particular, the dual objective function can be given by any of the four equivalent expressions:
\begin{align}
\ip{W_1+W_2}{J(\I_{\L(\calX)})} &= \ip{A}{\Tr_{\calX_2}^*(W_1)+ \Tr_{\calX_1}^*(W_2)} = \ip{A}{(\I \otimes \Phi_1^* \otimes \Phi_2^*)(\rho)} \\
&= \ip{J(\Phi_1 \odot_A \Phi_2)}{\rho}
\end{align}
where $A$ satisfies the conditions in \eqref{eq:Tr1ATr2A}. Note that they are equivalent since we are assuming $(W_1, W_2, \rho)$ is dual feasible.
\end{remark}

%%%%%%%%%%%%%%%%%%%%%%%%%%%%%%%
%%%%%%%%%%%%%%%%%%%%%%%%%%%%%%% 

\section{Qubit channels} \label{sec:qubitchanels}
\label{SectNumerical}

We now consider the case where $\calX =\calY= \complex^2$. From a result about the symmetric extendibility of two-qubit states~\cite{ChenJiKribsLutkenhausZeng-symExtension}, we have that a qubit channel $\Phi\in\C(\calX)$ is self-compatible if and only if it holds that
\begin{equation}
\label{eq:qubit-symmext}
\Tr\left(\bigl(\Tr_{\calX}(J(\Phi))\bigr)^2\right) \geq \Tr(J(\Phi)^2) - 4 \sqrt{\det(J(\Phi))}.
\end{equation}
However, closed-form algebraic criteria for the compatibility of other channels remain unknown.

\subsection{On the compatibility of (partially) dephasing-depolarizing channels}

In this subsection, we consider the compatibility and $k$-self-compatibility of a class of channels that we call \emph{dephasing-depolarizing channels}, which we define as follows. For fixed parameters ${p, q \in [0,1]}$ satisfying $p + q \leq 1$, the \emph{partially dephasing-depolarizing channel} is the linear map  
\begin{equation}
\label{eq:dephasingdepolarizing}
\Xi_{p,q} = (1 - p - q) \I_{\L(\calX)} + p \Delta + q \Omega,
\end{equation}
where we recall the (completely) dephasing channel $\Delta$ and the (completely) depolarizing channel $\Omega$ from Equation~\eqref{qubitchannels}.

\subsubsection*{Self-compatibility}

It can be easily verified that $\Xi_{p,q}$ is a measure-and-prepare channel (i.e., its Choi representation is separable, or in this case, PPT) if and only if $q\geq 2(1-p)/3$ (which is easy to check from the PPT criterion). Thus, for values of $p$ and $q$ in this region, the channel $\Xi_{p,q}$ is $k$-self-compatible for all $k$. We now consider the $k$-self-compatibility for these channels for small values of $k$.

The case $k=2$ is simple, as one may apply the condition \eqref{eq:qubit-symmext} to determine self-compatibility. Note that $\Tr_{\calX}(J(\Xi_{p,q}))= \I_{\calY}$ and thus $\Tr((\Tr_{\calX}(J(\Xi_{p,q})))^2) = 2$ holds for these channels. It is then straightforward to verify that \eqref{eq:qubit-symmext} is satisfied for $p,q\in[0,1]$ with $p+q\leq 1$ if and only if
\begin{equation}
\label{eq:qpineq}
q\geq \frac{2 - p - \sqrt{1 + 2p(1-p)}}{3}.
\end{equation}
This yields the region for which $\Xi_{p,q}$ is self-compatible. (Note that if one could not obtain a closed-form expression for the self-compatibility of a given family of channels analytically, it could be computed numerically by solving the relevant SDPs provided in Section~\ref{SectSDP}.)

For $k \in \{3, \dots 10\}$, we also numerically determine the regions where $\Xi_{p,q}$ is $k$-self-compatible using the SDP solver CVX \cite{cvx, GrantBoyd-convex}. The results are depicted in Figure~\ref{fig:Xipq_ksymmext}. Note that $k$-self-compatibility implies $l$-self-compatibility for $k \geq l$. (We note that although we did not give the SDP explicitly for $k$-self-compatibility, one can extend the SDP for self-compatibility in the natural way.  In fact, it is equivalent to the $k$-symmetric extendibility of the (normalized) Choi representation, and the SDP is then given in~\cite{DPS}.) 

\subsubsection*{Jordan self-compatibility}

Note that, for parameters $p, q > 0$ satisfying $p + q < 1$, the map $\Xi_{p, q}$ is invertible as a linear map. Indeed, the inverse map is given as
\begin{equation}
\label{eq:Xipq-inverse}
\Xi_{p, q}^{-1} = \dfrac{1}{1-p-q} \left( \I_{\L(\calX)} - \dfrac{p}{1-q} \Delta - q \dfrac{1-p-q}{1-q} \Omega \right).
\end{equation}
By Proposition \ref{thm:jordan-gen-inverseIFF}, and for these values of $p$ and $q$, the channel $\Xi_{p,q}$ is therefore self-compatible if and only if it is (generalized) Jordan compatible with itself. Recall that every measure-and-prepare channel is self-compatible. Thus, the region of pairs $(p,q)$ such that $\Xi_{p, q}$ is a measure-and-prepare channel (i.e., such that the Choi representation $J(\Xi_{p, q})$ is a separable operator) forms a subset of the region of pairs for which $\Xi_{p, q}$ is self-compatible. Consider now the standard Jordan product of $\Xi_{p, q}$ with itself. If the map $\Xi_{p, q}\odot\Xi_{p, q}$ is completely positive (i.e., if $J(\Xi_{p, q}\odot\Xi_{p, q})\geq0$), then $\Xi_{p, q}$ is self-compatible.

The boundaries of these three regions are depicted in Figure~\ref{fig:JP-XiPQ}. Interestingly, the region where $J(\Xi_{p, q}\odot\Xi_{p, q})\geq0$ is strictly smaller than the region where $\Xi_{p, q}$ is self-compatible and strictly larger than the region where $J(\Xi_{p, q})$ is separable.

\begin{figure}[t]
\centering
\begin{subfigure}[t]{.475\linewidth}
\centering
\includegraphics[width=\linewidth]{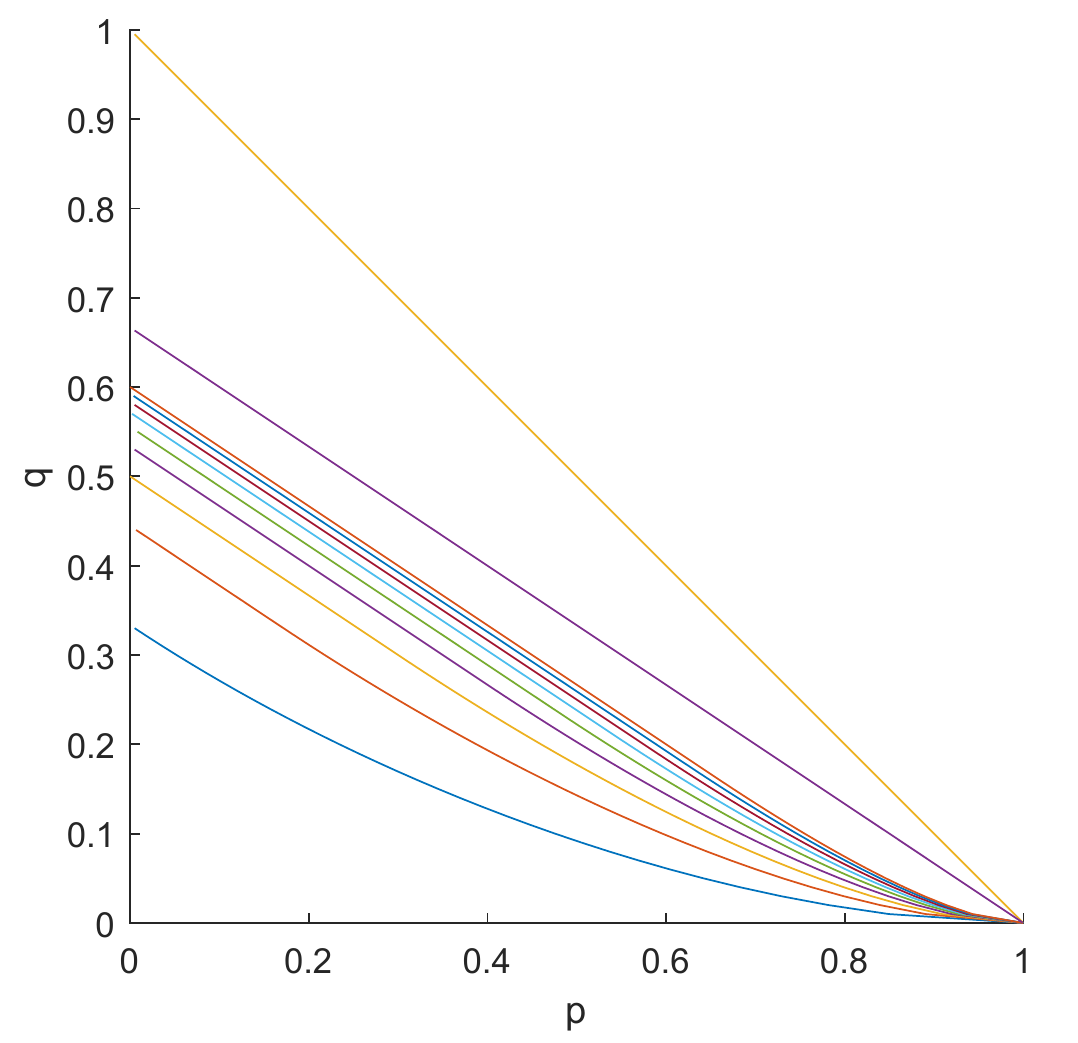}
\caption{Boundaries for the regions of $k$-self-compatible of the channel $\Xi_{p,q}$. The topmost (yellow) line is $p+q=1$ below which is the region in consideration. The next line down (purple) is defined by $q=2(1-p)/3$ and is the lower boundary for which $\Xi_{p,q}$ is $k$-self-compatible for all values of $k$. The bottom-most curve (blue) is the lower boundary of the $2$-self-compatible channels given by $q = \frac{2 - p - \sqrt{1 + 2p(1-p)}}{3}$. The subsequent curves above represent the boundaries of the $k$-self-compatible channels for $k \in \{3, \ldots, 10\}$ which were determined numerically (using CVX).}\label{fig:Xipq_ksymmext}
\end{subfigure}
\hfill
\begin{subfigure}[t]{.475\linewidth}
\centering
\includegraphics[width=\linewidth]{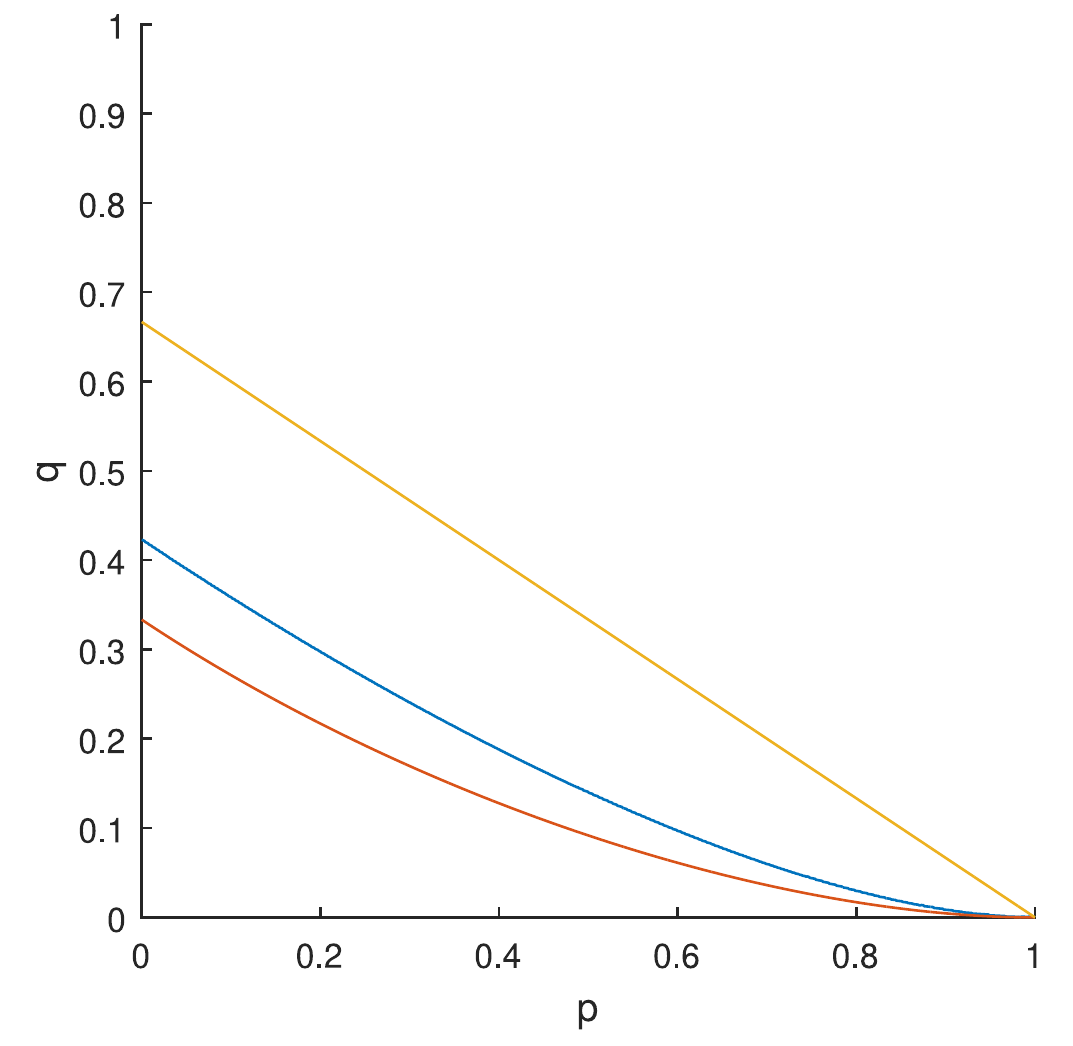}
\caption{The bottom-most curve (red) given by $q = \frac{2 - p - \sqrt{1 + 2p(1-p)}}{3}$ is the lower boundary of the region where the channel $\Xi_{p, q}$ is self-compatible. The middle curve (blue) is the lower boundary of the region where $J(\Xi_{p, q} \odot \Xi_{p, q}) \geq 0$, i.e., (standard) Jordan self-compatible. The top-most line (yellow) is the lower boundary of the region where $\Xi_{p, q}$ is a measure-and-prepare channel, i.e., $k$-self-compatible for all $k$. Recall that $\Xi_{p,q}$ is only defined for $p+q \leq 1$.}\label{fig:JP-XiPQ}
\end{subfigure}
\caption{{The regions of $k$-self-compatibility and (standard) Jordan self-compatibility of the partially dephasing-depolarizing channels.}}
\end{figure}

%%%%%%%%%%%%%%%%%%%%%%%%%%%%%%%%%%%%%%%%%%%%%%%%%%%%%%%%

\subsection{On the compatibility of partially depolarizing channels}
  
We now restrict our attention to pairs of partially depolarizing channels $(\Omega_{q_0}, \Omega_{q_1})$ as defined in~\eqref{eq:partialdepolarising}. (To connect this analysis to the discussion in the previous section, note channels of this form may be expressed as $\Omega_q = \Xi_{0,q}$.) It is evident that $\Omega = \Omega_1$ is compatible with every channel (see Lemma \ref{lemma:jordan-constant-NaS}). Moreover it is known that the channel $\Omega_{q}$ is self-compatible whenever ${1/3 \leq q \leq 1}$ (see \cite{Werner-cloning} for details). Furthermore, the set of pairs of compatible channels must be convex. Hence, the region of pairs $(q_0,q_1)$ for which $\Omega_{q_0}$ and $\Omega_{q_1}$ are compatible must at least contain the convex hull of these pairs. In fact, it was shown in \cite{Haapasalo-compat} that the channels $(\Omega_{q_0}, \Omega_{q_1})$ are compatible if and only if
\begin{equation}
q_0 + \sqrt{q_0 q_1} + q_1 \geq 1.
\end{equation}
We add to this analysis by considering pairs $(q_0, q_1)$ where $J(\Omega_{q_0} \odot \Omega_{q_1}) \geq 0$. (That is, when $\Omega_{q_0}$ and $\Omega_{q_1}$ are (standard) Jordan compatible.) To illustrate the distinction between these three different regions of compatibility, the boundaries of these regions are depicted in Figure~\ref{fig:JP-Deltaq0q1}.
\begin{figure}[t]
\centering
\includegraphics[width=.475\textwidth]{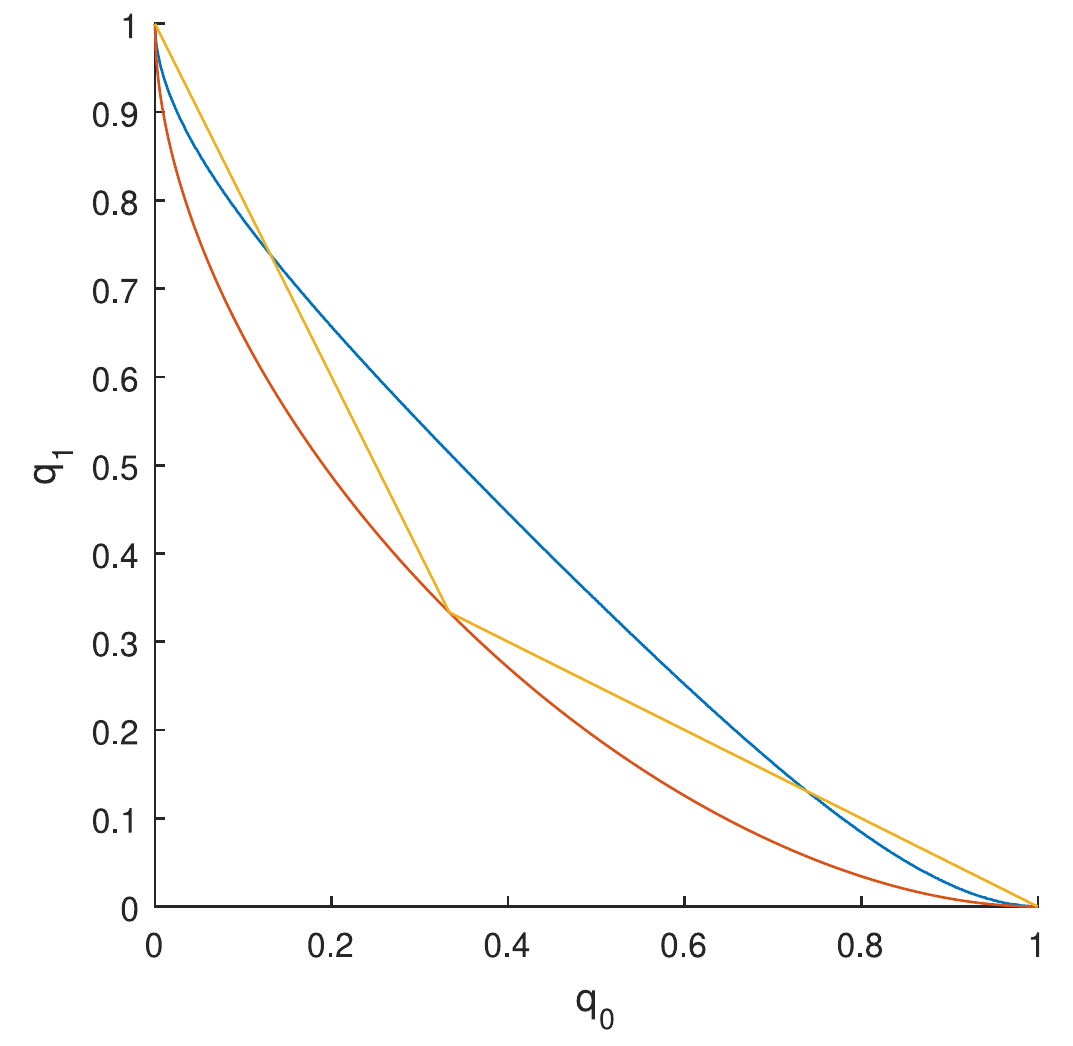}
\caption{Regions of points $(q_0,q_1)$ where the channels $\Omega_{q_0}$ and $\Omega_{q_1}$ satisfy each given property. The red curve is the lower boundary of the region where the channels $\Omega_{q_0}$ and $\Omega_{q_1}$ are compatible. The blue curve is the lower boundary of the region where $J(\Omega_{q_0} \odot \Omega_{q_1}) \geq 0$. The yellow curve is the lower boundary of the convex hull of the points $(1/3,1/3)$, $(0,1)$, $(1,0)$, and $(1,1)$.}
\label{fig:JP-Deltaq0q1}
\end{figure}
These are the boundaries of the regions of values $(q_0,q_1)$ where
\begin{itemize}
\item $\Omega_{q_0}$ and $\Omega_{q_1}$ are compatible (red curve),
\item $J(\Omega_{q_0} \odot \Omega_{q_1}) \geq 0$ (i.e., where $\Omega_{q_0}$ and $\Omega_{q_1}$ are (standard) Jordan compatible) (blue curve), and
\item $(\Omega_{q_0},\Omega_{q_1})$ is in the convex hull of the self-compatible pair $(\Omega_{1/3},\Omega_{1/3})$ and the trivially compatible pairs $(\Omega_{0},\Omega_{1})$, $(\Omega_{1},\Omega_{0})$, and $(\Omega_1, \Omega_1)$ (yellow curve)
\end{itemize}
(where we note that every channel is trivially compatible with the completely depolarizing channel $\Omega_1=\Omega$). Interestingly, the region of points $(q_0,q_1)$ satisfying $J(\Omega_{q_0} \odot \Omega_{q_1}) \geq 0$ is neither a subset nor a superset of the convex hull of the points $(1/3,1/3)$, $(0,1)$, $(1,0)$, and $(1,1)$. This illustrates that the (standard) Jordan product provides an interesting non-trivial sufficient condition for the compatibility of these channels.

Note that for values of $q$ satisfying $0 \leq q < 1$, the partially depolarizing map $\Omega_{q}$ is invertible with inverse map given by
\begin{equation}
\Omega_q^{-1} = \frac{1}{1-q} \left( \I_{\L(\calX)} - q \Omega \right).
\end{equation}
Therefore, for parameters satisfying $0 < q_0 < 1$ and $0 < q_1 < 1$, the channels $\Omega_{q_0}$ and $\Omega_{q_1}$ are compatible if and only if they are (generalized) Jordan compatible.

%%%%%%%%%%%%%%%%%%%%%%%%%%%%%%%
%%%%%%%%%%%%%%%%%%%%%%%%%%%%%%%

\section{Conclusions}
\label{SectConclusions}

In this work, we studied the quantum channel marginal problem (i.e., the channel compatibility problem)---which is the task of determining whether two channels can be executed simultaneously, in the sense that afterwards, one can choose which channel's output to obtain. We showed how to decide this via semidefinite programming and presented several other key properties such as its equivalence to the quantum state marginal problem.

We also studied a generalization of the Jordan product to quantum channels, and in turn, generalized it further such that it captures the compatibility of invertible channels. This Jordan product may be of independent interest.

%%%%%%%%%%%%%%%%%%%%%%%%%%%%%%%

\subsubsection*{Open problems.}
There are many open problems concerning the compatibility of channels. We briefly mention a few which we think are interesting.

One immediate open problem is the question of whether compatibility and Jordan compatibility are equivalent. We conjecture that they are equivalent based on the fact that the set of pairs of compatible channels that are not Jordan compatible must have zero measure when the output and input spaces have the same dimension. On this note, it would be interesting to see if the resolution to this conjecture depends on the dimensions of the input and output spaces.

Another interesting problem is to examine the computational complexity of determining compatibility for a given pair of channels. Since it is equivalent to the quantum state marginal problem, we suspect that there are versions of this problem which are QMA-hard (although the equivalence is a mathematical one, and may or may not translate into efficient algorithmic reductions).

Another open problem is whether one can extend this work to study the compatibility of other quantum objects, such as quantum strategies~\cite{GutoskiWatrous-QS, Gutoski-QS, GutoskiRosmanisSikora-QS}, combs~\cite{Ziman-PPOVM, ChiribellaDArianoPerinotti-combs}, or even channels in other generalized probabilistic theories.

Lastly, there might be a relationship between channel compatibility and cryptography. For example, symmetric extendibility is closely related to quantum key distribution, since you do not want Alice to be just as correlated/entangled with Bob and she is with Eve. Since quantum channel compatibility generalizes symmetric extendibility, perhaps there is another cryptographic setting in which the notion of channel compatibility translates into (in)security.

%%%%%%%%%%%%%%%%%%%%%%%%%%%%%%%
%%%%%%%%%%%%%%%%%%%%%%%%%%%%%%% 

\section*{Acknowledgements}
We thank John Watrous for helpful discussions and for coining the term ``compatibilizer''. JS also thanks Anurag Anshu and Daniel Gottesman for interesting discussions about the capacity of compatibilizing channels. MP is thankful to Teiko Heinosaari for discussing the Jordan product of channels and the compatibility of measure-and-prepare channels.

MG is supported by the Natural Sciences and Engineering Research Council (NSERC) of Cana\-{}da, the Canadian Institute for Advanced Research (CIFAR), and through funding provided to IQC by the Government of Canada.

MP is thankful for the support by Grant VEGA 2/0142/20, by the grant of the Slovak Research and Development Agency under Contract No. APVV-16-0073, by the DFG and by the ERC (Consolidator Grant 683107/TempoQ).

JS is supported in part by the Natural Sciences and Engineering Research Council (NSERC) of Cana\-{}da. 

Research at Perimeter Institute is supported in part by the Government of Canada through the Department of Innovation, Science and Economic Development Canada and by the Province of Ontario through the Ministry of Colleges and Universities. 
   
%%%%%%%%%%%%%%%%%%%%%%%%%%%%%%% 
\bibliographystyle{ieeetr}
\bibliography{citations}

\end{document}